\documentclass[journal]{IEEEtran}
\usepackage{amsfonts}
\usepackage{amssymb}
\usepackage[cmex10]{amsmath}
\usepackage{graphicx}
\usepackage{subfigure}
\usepackage{cite}
\usepackage{tabularx,booktabs}
\usepackage{amsthm}
\usepackage{booktabs}
\usepackage[ruled,lined]{algorithm2e}
\usepackage{multirow}
\usepackage{setspace}
\usepackage{fancyhdr}
\usepackage{algorithmic}
 \usepackage{enumerate}
 \usepackage{color}

\IEEEoverridecommandlockouts
\newcommand{\RNum}[1]{\uppercase\expandafter{\romannumeral #1\relax}}
\newcommand{\be}{\begin{equation}}
\newcommand{\ee}{\end{equation}}

\newcommand{\bII}{{\bf I }}

\newcommand{\bR}{R}

\newcommand{\bx}{x}

\newcommand{\bH}{{H}}

\newcommand{\bF}{{F}}

\newcommand{\bQ}{{Q}}

\newtheorem{Def}{Definition}
\newtheorem{Pro}{Proposition}

\newtheorem{Rem}{Remark}
\newtheorem{Lem}{Lemma}

\allowdisplaybreaks
\ifCLASSINFOpdf
\else
\fi

\begin{document}

\title{Computationally Efficient Distributed Multi-sensor Fusion with  Multi-Bernoulli  Filter}
\author{
Wei Yi, Suqi Li$^\ast$, Bailu Wang, Reza Hoseinnezhad and Lingjiang Kong
\thanks{W. Yi, S. Li, B. Wang and L. Kong are with the School of Information and Communication Engineering, the University of Electronic Science and Technology of China. \emph{(Suqi Li is the corresponding author, Email:qi\_qi\_zhu1210@163.com)}

R. Hoseinnezhad is with the School of Engineering, RMIT University,
Victoria 3083, Australia.
}}
\maketitle
 \thispagestyle{empty}
\begin{abstract}
This paper proposes a computationally efficient algorithm for distributed fusion in a sensor network in which multi-Bernoulli (MB) filters are locally running in every sensor node for multi-target tracking. The generalized Covariance Intersection (GCI) fusion rule is employed to fuse multiple MB random finite set densities. The fused density comprises a set of fusion hypotheses that grow exponentially with the number of Bernoulli components. Thus, GCI fusion with MB filters can become computationally intractable in practical applications that involve tracking of even a moderate number of objects. In order to accelerate the multi-sensor fusion procedure, we derive a theoretically sound approximation to the fused density. The number of fusion hypotheses in the resulting density is significantly smaller than the original fused density. It also has a parallelizable structure that allows multiple clusters of Bernoulli components to be fused independently. By carefully clustering Bernoulli components into isolated clusters using the GCI divergence as the distance metric, we propose an alternative to build exactly the approximated density without exhaustively computing all the fusion hypotheses. The combination of the proposed approximation technique and the fast clustering algorithm can enable a novel and fast GCI-MB fusion implementation. Our analysis shows that the proposed fusion method can dramatically reduce the computational and memory requirements with small bounded $L_1$-error. The Gaussian mixture implementation of the proposed method is also presented. In various numerical experiments, including a challenging scenario with up to forty objects, the efficacy of the proposed fusion method is demonstrated.
\end{abstract}


%
\IEEEpeerreviewmaketitle
\section{Introduction}
\IEEEPARstart{D}{istributed} Multi-sensor Multi-object Tracking (DMMT) technology is recently advocated in the information fusion community, since it generally benefits from lower communication costs and immunity to single-node fault, compared with optimal centralized fusion solutions. In practice, the correlations between the posterior densities from different sensors are usually unknown, and this makes it very challenging to devise DMMT solutions.  The optimal solution to this problem has been proposed by Chong~et~al.~\cite{CY-Chong}, but the heavy computational burden of extracting the common information can rule out this solution in many applications. A suboptimal solution with demonstrated tractability has been formulated based on the {Generalized Covariance Intersection} (GCI)
method proposed by Mahler \cite{Mahler-1}. Essentially, GCI is the generalization of {Covariance Intersection} (CI) \cite{Uhlmann}. Compared to the CI fusion rule, which mainly utilizes the first and second order statistical characteristics of the single-object densities, the GCI fusion rule is capable to fuse multiple full multi-object densities with unknown correlations among sensor nodes, while intrinsically avoiding the double counting of common information \cite{double-counting}.
The GCI fused multi-object density is also named as {Exponential Mixture Density}~(EMD)~\cite{Uney-2,EMD-Julier} or {Kullback-Leibler Average}~(KLA)~\cite{Battistelli}.

Clark~et~al.~\cite{Clark} derived tractable implementations of the GCI fusion rule for especial cases where the multi-object distributions are Poisson, independent identically distributed (i.i.d.) clusters and Bernoulli distributions. Utilizing these formulations, the GCI fusion method can be implemented for distributed fusion in sensor networks where the Probability Hypothesis Density (PHD)/Cardinalized PHD (CPHD) filter or the Bernoulli filter are running locally in every sensor node. Battistelli~et~al.~\cite{Battistelli} and Guldogan~\cite{Mehmet} proposed the Gaussian Mixture (GM) implementation of distributed fusion with CPHD filter and Bernoulli filter, respectively. Furthermore, a
consensus approach to distributed multi-object tracking was first introduced in \cite{Battistelli}, with which the GCI fusion can be implemented in a fully distributed manner.  The Sequential Monte Carlo (SMC) implementation of the GCI fusion with CPHD filter was proposed by Uney~et~al.~\cite{Uney-2}.

Compared with PHD \cite{PHD-Vo,book_mahler} and CPHD~\cite{refr:CPHD,Franken,Vo-CPHD,MeMber_Vo2} filters, the Multi-Bernoulli (MB) filter \cite{MeMber_Vo2,MeMber_Vo3} can be more efficient and accurate in problems that require object individual existence probabilities.
The performance of the MB filters has been well demonstrated in a wide range of practical monitoring problems such as radar tracking \cite{Radar-tracking}, video and audio tracking \cite{MeMber_Vo3, Reza_visual_tracking}, sensor control \cite{Reza_sensor_control_AES, Reza_sensor_control_letter}, and acoustic sensor tracking \cite{MeM-superpositional-sensor}.
In addition, a number of novel extensions of the MB recursion have also been investigated. Examples include the multiple model multi-Bernoulli filter proposed to cope with maneuvering objects~\cite{Dunne_MM}, a robust multi-Bernoulli filter that can accommodate an unknown non-homogeneous clutter and detection profile~\cite{Vo_radar_target_nonline}, a joint multi-object estimator for multi-Bernoulli models~\cite{Estimator-Multi-Bernoulli}, and a few enhanced multi-Bernoulli filters~\cite{Enhanced-multi-Bernoulli-1,Enhanced-multi-Bernoulli-2}.

The generalization of MB filters to the DMMT problem was firstly investigated in~\cite{GCI-MB}, where a distributed fusion algorithm based on GCI fusion rule (called GCI-MB in this paper) is proposed. While the enhanced performance of the GCI-MB fusion has been demonstrated~\cite{GCI-MB}, its major drawback is the heavy computational cost that increases exponentially with the number of objects in the surveillance area. For the DMMT over a sensor network, it is of paramount importance to reduce the local (in-node) computational expenses as much as possible, since in the typical distributed sensor networks, each node usually has limited processing power and energy resources.

In this paper, the computations involved in GCI-MB fusion are investigated, showing that the major contributor to the heavy computational burden is the exhaustive calculation of the weights and fused single-object densities of the GCI-MB fusion hypotheses, noting that the number of hypotheses increases super-exponentially with the number of objects. This observation provides the insight that leads to a novel computationally efficient and naturally parallelizable implementation of the GCI-MB fusion.

The major contributions are summarized as follows:
\begin{enumerate}
  \item \textit{Devising a principled approximation to the fused density}: By discarding all the insignificant GCI-MB fusion hypotheses with negligible weights, then normalizing the rest of hypotheses, we obtain a principled approximation to the original fused density. This approximation not only consists of significantly smaller number of fusion hypotheses but also enjoys an appealing structure: it can be factorized into several mutually independent and smaller size multi-object densities that lead to a principled independence approximation. Each of these factorized densities are shown to be equally returned by the GCI-MB fusion performed with a partial cluster of entire Bernoulli components. We also quantify the approximation error by deriving the $L_1$-error~\cite{LMB_Vo2} between the original  density and the approximated one, and show that it is reasonably small, and indeed, negligible for many practical applications.
   \item \textit{Fast clustering of the Bernoulli components}: The main challenge in obtaining the approximation is to truncate all the insignificant GCI-MB fusion hypotheses without computing them all in an exhaustive manner. To this end, by carefully clustering the Bernoulli components into isolated clusters according to a distance metric, i.e., GCI divergence, and then discarding the hypotheses of the association pairs of Bernoulli components from different isolated clusters among sensors, one can build exactly the independence approximation without exhaustively computing all the terms. By modeling the underline data structure of Bernoulli components among all sensors as a undirected graph, the clustering problem is tantamount to seeking connected components of a undirected graph dynamically, and can be solved by the disjoint-set data structure based fast solver with a computational cost that is only polynomial in the number of Bernoulli components.
  \item \textit{Computationally efficient GCI-MB fusion}: The combination of the independence approximation and the fast clustering algorithm can enable a computationally efficient fusion algorithm. The fast clustering algorithm is employed at an early stage to obtain the approximated density. Then by utilizing the structure of the approximated density, GCI fusion is performed within each smaller clusters of Bernoulli components independently and in a parallelized manner, namely, via performing fusion between each pair of the factorized multi-object densities among sensors. The GM implementation of the proposed fast GCI-MB fusion is also presented.
\end{enumerate}
In numerical results, the efficiency and accuracy of the proposed fast GCI-MB fusion algorithm with GM implementation is tested in a challenging scenario with up to forty objects. Preliminary results have been published in the conference paper \cite{conference-version}. This paper presents a more complete theoretical and numerical study.

The rest of the paper is organized as follows. Background material and notation are introduced in Section II.
The GCI-MB fusion is briefly reviewed and its computational intractability is discussed in Section III.
Section IV devotes to the development of the proposed distributed fusion algorithm with MB filter based on the independence approximation and the fast clustering algorithm. Section V presents the GM implementation of the proposed GCI-MB fusion algorithm, including the algorithm pseudocode, and the analysis of its computational complexity. Section VI demonstrates the effectiveness of the proposed fusion algorithm via various numerical experiments. Concluding remarks are finally provided in Section VII.

\section{Background and Notation}
Consider a sensor network which is constructed by a set of sensor nodes, $s=1,\cdots,N_s$. Each sensor node is linked with its neighbours to exchange information, and is assumed to be equipped with some limited processing, memory and other electronic hardware for sensing, signal processing and communication.

To describe the statistics of unknown and time-varying number of objects, the multi-object state at time $k$ is naturally represented as a \textit{Random Finite Set} (RFS) \cite{book_mahler,book2_mahler} $X^k=\{\bx^k_1,\bx^k_2,\ldots,\bx^k_{n}\}\in \mathcal{F}(\mathbb{X})$, where $\mathbb{X}$ is the space of single-object states, and $\mathcal{F}(\mathbb{X})$ is the collection of all finite subsets of $\mathbb{X}$. Also, $\mathcal{F}_n(\mathbb{X})$ denotes the collection of all finite subsets of $\mathbb{X}$ with the cardinality $n$. At each time step $k$, each sensor node $s$ receives the observations of multiple objects, denoted by $Z_s^{k}=\{z^k_{1,s},\cdots,z^k_{m,s}\}\in\mathcal{F}(\mathbb{Z})$, where $\mathbb{Z}$ is the space of observations, and $\mathcal{F}(\mathbb{Z})$ is the collection of all finite subsets of $\mathbb{Z}$. Let $Z_s^{1:k}=(Z_s^{1},\ldots,Z_s^{k})$ denotes the history of observations from time $1$ to time $k$ at sensor $s$.

\subsection{Distributed multi-object tracking}
DMMT for a sensor network usually consists of three  stages  described as follow:

--\textit{ Local  filtering}: Each node $s$ collects observations of multiple objects $Z_s^k$ at each time step $k$, and performs local filtering based on the multi-object Bayesian recursion \cite{book_mahler}.

-- \textit{Information exchange:} Each sensor node exchanges the multi-object posteriors with its neighbours via communication links.

-- \textit{Fusion of posteriors:} Each node performs a fusion operation to combine its locally computed posterior with the ones communicated by its neighbours.

\subsection{Local filtering using multi-Bernoulli filters}
\subsubsection{Multi-Bernoulli distribution}
An RFS $X$ distributed according to an MB distribution is defined as the union of $M$ independent Bernoulli RFSs  $X^{(\ell)}$~\cite{Mahler-1},
\begin{equation}
X=\bigcup_{\ell=1}^{M}X^{(\ell)}.
\end{equation}
The MB distribution is completely characterized by a set of parameters $\{(r^{(\ell)},p^{(\ell)})\}_{\ell=1}^{M}$, where $r^{(\ell)}$ denotes the existence probability and $p^{(\ell)}(\cdot)$ denotes the probability density of the $\ell$-th Bernoulli RFS, or Bernoulli component. The multi-object probability density of an MB RFS is given by \cite{Mahler-1},
\begin{align}\label{multi-Bernoulli}
\begin{split}
&\pi(\left\{x_1,\ldots,x_n\right\})=\sum_{1\leq i^1\neq\ldots\neq i^n\leq M}Q^{(i^1,\cdots,i^n)}\prod_{j=1}^n p^{(i^{j})}(\bx_{j})
\end{split}
\end{align}
where $n=\{1,\cdots,M\}$ and
\begin{align}\label{Phi}
\begin{split}
Q^{(i^1,\cdots,i^n)}=\prod_{\ell=1}^{M}(1-r^{(\ell)})\prod_{j=1}^n\frac{r^{(i^{j})}}{1-r^{(i^{j})}}.
\end{split}
\end{align}
In addition to (\ref{multi-Bernoulli}), another equivalent  form of MB distribution could be expressed as \cite{GCI-MB}
\begin{equation}\label{multi-Bernoulli_V2}
\begin{split}
\pi&(\!\left\{\bx_1\!,\!\ldots\!,\!\bx_n\right\}\!)
=\sum_{\sigma}\sum_{I\in \mathcal{F}_n(\mathbb{L})}Q^{I}\prod_{i=1}^{n}p^{([I]^v(i))}(\bx_{\sigma(i)})
\end{split}
\end{equation}
where
$\mathbb{L}\!\triangleq\!\{1,\ldots,M\}$ is a set of indexes numbering Bernoulli components, $I$ is the  set of indexes of densities,
$\sigma$ denotes one possible permutation of $I$, $\sigma(i)$ denotes the {$i$-th} element of the permutation,  the summation $\sum_\sigma$ is taken over all permutations on the numbers $\{1,\cdots,n\}$,  $[I]^v$ is a vector constructed by sorting the elements of the set $I$, and
 \begin{align}\label{Phi2_1}
\begin{split}
&Q^{I}=\prod_{\ell'\in I}{r^{(\ell')}}\prod_{\ell\in \mathbb{L}\backslash I}(1-r^{(\ell)}).
\end{split}
\end{align}
Hereafter, each Bernoulli component $\ell\in\mathbb{L}$, paramerized by $(r^{(\ell)},p^{(\ell)}(\cdot))$, is also called as a hypothetic object.
This naming is natural and intuitive  in the sense that $r^{{(\ell)}}$ describes the probability of existence of the hypothetic object, while  $p^{(\ell)}(\cdot)$ describes the probability distribution of the hypothetic object conditioned on existence.
 Note that the subsequent results of this paper follow the MB distribution of form (\ref{multi-Bernoulli_V2}).
\subsubsection{Multi-Bernoulli filtering}
An MB filter recursively computes and propagates an MB posterior forwards in time, via Bayesian prediction and update steps. Under the standard motion model, the MB density is closed for the Chapman-Kolmogorov equation  \cite{book_mahler}. However, the multi-object posterior resulted from the update step of the MB filter is not necessarily an MB posterior and depends on the observation models. The standard observation model, image model, and superpositional sensor model lead to different forms of multi-object posteriors that are approximated by different variants of MB densities \cite{MeMber_Vo2,MeMber_Vo3,MeM-superpositional-sensor}.

\subsection{Information exchange and fusion based on GCI}
At time $k$, assume two nodes 1 and 2 in a sensor network maintain their local posteriors $\pi_{1}(X^k|Z_{1}^{1:k})$ and $\pi_{2}(X^k|Z_{2}^{1:k})$ which are both RFS multi-object densities. Each sensor node exchanges its local multi-object posterior with the other. The fused posterior based on GCI fusion rule is the geometric mean, or the exponential mixture of the local posteriors,
\begin{align}\label{G-CI}
\begin{split}
\!\!\!\pi_{\omega}(X^k|Z_1^{1:k},Z_{2}^{1:k})\!=\!\frac{\pi_{1}(X^k|Z_{1}^{1:k})^{\omega_1}\pi_{2}(X^k|Z_{2}^{1:k})^{\omega_2}}
                                              {\int \pi_{1}(X^k|Z_{1}^{1:k})^{\omega_1}\pi_{2}(X^k|Z_{2}^{1:k})^{\omega_2}\delta X}
\end{split}
\end{align}
where $\omega_1$ and $\omega_2$ represent the relative fusion weights of each nodes. The weights are normalized, i.e. $\omega_1+\omega_2=1$. One possible choice of the weights  to ensure converge of the fusion is the so-called Metropolis weights \cite{Xiao, consensus-weight-2, Battistelli}. An alternative is based on an optimization process in which the objective weights minimize a selected cost function, such as the determinant or the trace of the covariance of the fused density (see \cite{Uney-2, EMD-Julier, adaptive-weight} for more  detailed  discussions). For convenience of notation, in what follows we omit explicit references to the time index $k$, and the condition $Z_s^{1:k}$.

The fused posterior given by (\ref{G-CI}) minimizes the weighted sum of its {Kullback-Leibler Divergence} (KLD) \cite{Battistelli} with respect to two given distributions,
\begin{equation}\label{EMD}
\begin{split}
  \pi_\omega=\arg \min_\pi(\omega_1D_{\text{KL}}(\pi\parallel \pi_1)+\omega_2 D_{\text{KL}}(\pi\parallel \pi_2))
  \end{split}
\end{equation}
where $D_{\text{KL}}$ denotes the KLD defined as
\begin{equation}\label{KLD}
\begin{split}
 D_{\text{KL}}(f||g)\triangleq \int f(X)\log{\frac{f(X)}{g(X)}}\delta X
  \end{split}
\end{equation}
where the integral in (\ref{KLD}) admits both the set integral \cite{book_mahler} and Euclidean notion of integral.

\begin{Rem}
The GCI fusion can be easily extended to $N_s>2$ sensors by sequentially applying (\ref{G-CI}) $N_s-1$ times, where the ordering of pairwise fusions is irrelevant. A similar approach has been used in distributed fusion, for instance in the GCI fusion with CPHD filters \cite{Battistelli} and MB filters \cite{GCI-MB}.
 \end{Rem}

\section{GCI-MB Fusion and Its Computational Intractability}
This section presents a brief review of the GCI-MB fusion  algorithm in \cite{GCI-MB}, then discusses its  computational intractability.
\subsection{GCI fusion of MB distributions}
Let the local multi-object density $\pi_s$, $s=1,2$ be an MB density parameterized by \begin{equation}
\pi_s=\{(r_s^{(\ell)},p_s^{(\ell)}(\cdot))\}_{\ell=1}^{M_s}.
\end{equation}

Consistent with the second form of MB distribution given in  (\ref{multi-Bernoulli_V2}), we define
\begin{equation}
\mathbb{L}_{s}=\{1,\cdots,M_{s}\}
\end{equation}
which is a set of indexes numbering Bernoulli components of sensor $s$, $s=1,2$.

For the subsequent development, a definition of the fusion map, which describes a hypothesis that a set of hypothetic objects in sensor 2 are one-to-one matching with a set of hypothetic objects in sensor 1, is provided firstly.

\begin{Def}\label{defination_fusion_maps}
Without loss of generality, assume that $|\mathbb{L}_1|\leq|\mathbb{L}_2|$. A fusion map (for the current time) is a function
$\theta: \mathbb{L}_1\!\rightarrow \!\mathbb{L}_2$ such that  $\theta(i)\!=\!\theta(i^{\ast})$ implies $i\!=\!i^{\ast}$.
 The set of all such fusion maps is called fusion map space denoted by $\Theta$. The subset of $\Theta$ with domain $I$ is denoted by $\Theta(I)$. For notational convenience, we define $\theta(I)\triangleq
\{\theta(i), i\in I\}$.
\end{Def}

According to Propositions 2-3 in \cite{GCI-MB}, the GCI-MB fusion contains the following two steps:

\noindent \textbf{Step 1} - \textit{Calculation of the fused density}: 
The GCI fusion of two MB densities in form of (\ref{multi-Bernoulli_V2}) yields  a generalized multi-Bernoulli (GMB) density \footnote{Note that the GMB density can be viewed as the unlabeled version of the generalized labeled multi-Bernoulli (GLMB) distribution \cite{LMB_Vo,LMB_Vo2}.} \cite{GCI-MB} of the following form,
\begin{equation}\label{GCI-MB}
\begin{split}
\!\pi_\omega(\left\{\bx_1,\!\cdots\!,\bx_n\right\})
\!\!=\!\!\sum_\sigma\!\!\!\sum_{(I_1,\theta)\in \mathcal{F}_n(\mathbb{L}_1)\!\times\! \Theta{(I_1)}}\!\!\!\!\!\!\!\!\!w_\omega^{(I_1,\theta)}\prod_{i=1}^{n}p_\omega^{([I_1]^v(i),\theta)}(\!\bx_{\sigma(i)}\!)
\end{split}
\end{equation}
where
\begin{align}
\label{fuse-w}
w_\omega^{(I_1,\theta)}&=\frac{\widetilde{w}_\omega^{(I_1,\theta)}}{\eta_\omega}\\
\label{fuse-p_p}p_{\omega}^{(\ell,\theta)}(\bx)&=\frac{p_1^{(\ell)}(\bx)^{\omega_1}p_2^{(\theta(\ell))}(\bx)^{\omega_2}}{Z_\omega^{(\ell,\theta)}},\,\,\,\,\ell\in I_1
\end{align}
with
\begin{align}
\label{fuse-w_p}\widetilde{w}_\omega^{(I_1,\theta)}&=\left(Q_1^{I_1}\right)^{\omega_1}\left(Q_2^{\theta(I_1)}\right)^{\omega_2}\prod_{\ell\in I_1}Z_\omega^{(\ell,\theta)}\\
\label{K-1}\eta_\omega&=\sum_{I_1\in \mathcal{F}(\mathbb{L}_1)}\sum_{\theta\in\Theta{(I_1)}}\widetilde{w}_\omega^{(I_1,\theta)}\\
\label{Z_w} Z_\omega^{(\ell,\theta)}&=\int p_1^{(\ell)}({\bx})^{\omega_1}p_2^{(\theta(\ell))}({\bx})^{\omega_2}d{\bx}.
\end{align}
It can be seen from (\ref{GCI-MB})-(\ref{Z_w}) that the fusion process is simple and intuitive. The fused GMB density of (\ref{GCI-MB}) is a mixture of multi-object exponentials. The fused single-object density $p_{\omega}^{(\ell,\theta)}(\bx)$ of (\ref{fuse-p_p}) can be viewed as the GCI fusion result of the single-object densities of two paired hypothetic objects from sensors 1 and 2. 
The quantity $\eta_\omega$ in (\ref{K-1}) is a normalization constant, and the following un-normalized fused density is referred to as the un-normalized GMB density hereafter,
\begin{align}\label{unnormalized}
\begin{split}
&\widetilde\pi_{\omega}(\{\bx_{1},\!\cdots\!,\bx_{n}\})\!\\
&=\!\sum_{\sigma}\!\!\sum_{(I_1,\theta)\in \mathcal{F}_n(\mathbb{L}_1) \times\Theta{(I_1)}}\widetilde {w}_\omega^{(I_1,\theta)}\prod_{i=1}^{n}p_\omega^{([I_1]^v(i),\theta)}(\bx_{\sigma(i)}).
\end{split}
\end{align}

\noindent \textbf{Step 2} -  \textit{MB approximation of the GMB density}:
To allow  the subsequent fusion with another MB density, the fused density should be also in the MB form. To this end, $\pi_{\omega}(X)$ is approximated by an MB distribution, {$\pi_{\omega,\text{MB}}(X)=\{(r_\omega^{(\ell)},p_\omega^{(\ell)})\}_{\ell\in\mathbb{L}_1}$,} which matches exactly its first-order moment, where
\begin{equation}
\begin{split}
\label{r_l}\!\!\!\!\!\!\!\!\!\!\!\!\!\!\!\!\!\!\!\!\!\!\!\!\!\!\!\!\!\!r_\omega^{(\ell)}&=\sum_{I_1\in \mathcal{F}(\mathbb{L}_1)}\sum_{\theta\in\Theta(I_1)}1_{I_1}(\ell)w_\omega^{(I_1,\theta)}
\end{split}
\end{equation}
\begin{equation}
\begin{split}
\label{p_l}p_\omega^{(\ell)}(x)&=\!\sum_{I_1\in \mathcal{F}(\mathbb{L}_1\!)}\sum_{\theta\in\Theta(I_1)}1_{I_1}(\ell)w_\omega^{(I_1,\theta)}p_{\omega}^{(\ell,\theta)}(x)\big{/}r_\omega^{(\ell)},
\end{split}
\end{equation}
where $1_{I_1}(\ell)=1$ if $\ell\in I_1$; otherwise, $1_{I_1}(\ell)=0$.

It is remarked  that the first-order moment matching  approximation  adopted here is a widely used approximation technique in the RFS based multi-object tracking algorithms  \cite{LMB-Reuter, MeMber_Vo2,Efficient-GCI-GLMB}.
\subsection{Computational intractability of the GCI-MB fusion}
Observing (\ref{GCI-MB}), the calculation of the fused GMB density involves a weighted sum of products of single-object densities. Hereafter, each term in the summation, indexed by $(I_1,\theta)\in\mathcal{F}(\mathbb{L}_1)\times\Theta$ for the fused GMB density is called as a \textit{fusion hypothesis}. Hence direct implementation of GCI-MB fusion needs to exhaust all fusion hypotheses and compute their corresponding weights and fused single-object densities.
\begin{Rem}
Observing (\ref{GCI-MB}) -- (\ref{Z_w}), calculation of each fusion hypothesis $(I_1,\theta)$ implicitly means that given the existence of the set of hypothetic objects $I_1$, any hypothetic object $\ell\in I_1$ of sensor 1 and the associated  hypothetic object $\theta(\ell)$  of sensor 2 are considered to be originated from the same object, and the respective statistical information should be fused based on~(\ref{fuse-p_p}).
\end{Rem}
According to Binomial theorem, the cardinality of $\mathcal{F}(\mathbb{L}_1)$ is
\begin{equation}\label{number_subsets}
|\mathcal{F}(\mathbb{L}_1)|={\sum}_{n=0}^{|\mathbb{L}_1|} C_{|\mathbb{L}_1|}^n=2^{|\mathbb{L}_1|},
\end{equation}
where $C_m^n$ denotes the number of $n$-combinations from a set of $m$ elements. The cardinality of the set $\Theta$ of fusion maps is
\begin{equation}\label{number_mappings}
|\Theta|=A_{|\mathbb{L}_2|}^{|\mathbb{L}_1|}.
\end{equation}
where $A_m^n$ denotes the number of $n$-permutations of $m$ elements. Consequently, the total number of fusion hypotheses, denoted by $N_H$, is computed by
\begin{equation}\label{number_hypothesis}
N_H=|\mathcal{F}(\mathbb{L}_1)\times\Theta|=2^{|\mathbb{L}_1|}\times A_{|\mathbb{L}_2|}^{|\mathbb{L}_1|} \geq 2^{|\mathbb{L}_1|}\times|\mathbb{L}_1|!.
\end{equation}
It can be seen from (\ref{number_hypothesis}) that the number of fusion hypotheses grows as the number of local Bernoulli components with a speed of at least $\mathcal{O}( 2^{|\mathbb{L}_1|}\times|\mathbb{L}_1|!)$. Note that due to the factorial term, the computational complexity grows super-exponentially with the number of local Bernoulli components, $|\mathbb{L}_1|$.

The number of local Bernoulli components is directly related to the implementation of the local MB filter. Theoretically, this number increases linearly with time steps $k$ (with no bound) due to the inclusion of of birth components at each time step. In practice, even in presence of a pruning strategy (to curb the growing number of Bernoulli components), this number can be significantly larger than the true number of objects.

The super-exponential rate of growth of number of hypothesis with the number of Bernoulli components which itself grows with the number of existing objects, together make the original GCI-MB fusion computationally stringent in many practical applications, especially those involving tracking of numerous objects. Thus, it is of practical importance to devise an efficient implementation of GCI-MB fusion algorithm.

\section{Computationally Efficient GCI-MB Fusion}
This section firstly provides an intuitive perception of the GCI-MB fusion through the analysis of ``bad'' association pair. Motivated by this analysis, we advocate a principled independence approximation to the original GMB density, and also characterize its approximation error in terms of the $L_{1}$-error. Finally, by utilizing this independence approximation, a computationally efficient GCI-MB fusion algorithm is developed.
\subsection{``Bad'' association pair}

\begin{Def}
An ordered pair that is comprised of hypothesized objects $\ell$ from sensor 1 and $\ell'$ from sensor 2, is called an association pair of hypothetic objects (or an association pair for short), and denoted by $(\ell,\ell')\in\mathbb{L}_1\times\mathbb{L}_2$.
\end{Def}

As we discussed in Remark 1, a fusion hypothesis $(I,\theta)$ implicitly means that under the existence of the set of hypothetic objects $I_{1}$ for sensor 1, each hypothetic object $\ell\in I_{1}$ is considered to be associated with the hypothetic object  $\ell'=\theta(\ell)$ of sensor 2. In this respect, a fusion hypothesis $(I_1,\theta)$ can be interpreted as a set of association pairs. With a little notational abuse, in the rest of the paper, we adopt the following notation,  \begin{equation}(I_1,\theta)\triangleq\{(\ell,\theta(\ell)):\ell\in I_1\}.\end{equation}

Intuitively, if statistical information relevant to the states of two hypothetic objects in an association pair have a large discrepancy, there is a small chance that they describe the same object and hence, the corresponding fusion is not well-posed. As a result, this association pair can be considered as a ``bad'' association pair.

In the context of the GCI fusion rule, we choose the GCI divergence \cite{GCI-GLMB} between location densities to quantify the discrepancy between the statistics of the corresponding hypothetic objects. For any two hypothetic objects $\ell$ and $\ell'$ with location densities $p_1^{(\ell)}(x)$ and $p_2^{(\ell')}(x)$, respectively, the GCI divergence between them is computed by
 \begin{equation}\label{GCI-distance}
d(\ell,\ell')= -\ln\int \left[p_1^{(\ell)}(x)\right]^{\omega_1}\left[p_2^{(\ell')}(x)\right]^{\omega_2}d x.
\end{equation}

\begin{Rem}
GCI divergence was first proposed in~\cite{GCI-GLMB} as a tool to quantify the degree of similarity/difference between statistical densities. It is an extension of the Bhattacharyya distance which is a well-known measure for the amount of overlap between two statistical samples or populations~\cite{Bhattacharyya-distance}. In the special case with $\omega_1=\omega_2=\frac{1}{2}$, the GCI divergence defined in~\eqref{GCI-distance} returns the Bhattacharyya distance.
\end{Rem}
\begin{Rem}
It has been demonstrated that a large GCI divergence between two densities, leads to the GCI fusion being significantly violating the Principle of Minimum Discrimination Information~\cite{GCI-GLMB}. In the extreme case where the GCI divergence approaches $+\infty$, the densities are not compatible from the Bayesian point of view because their supports are disjoint, and the GCI fusion is not well defined.
\end{Rem}
Utilizing the GCI divergence as a distance between hypothetic objects, we present a measurable definition of a ``bad'' association pair.
 \begin{Def}\label{bad-association-pair}
 Given an association pair $(\ell,\ell')\in\mathbb{L}_1\times\mathbb{L}_2$, if the distance between hypothetic objects $\ell$ and $\ell'$ satisfies
\begin{equation}\label{GCI-distance-threshold}
d(\ell,\ell')=-\ln\int \left[p_1^{(\ell)}(x)\right]^{\omega_1}\left[p_2^{(\ell')}(x)\right]^{\omega_2}d x>\gamma,
\end{equation}
$(\ell,\ell')$ is said to be a ``bad'' association pair, where $\gamma$ is a predefined sufficiently large threshold.
\end{Def}

Note that in GCI-MB fusion equations~(\ref{fuse-w_p})~and~(\ref{Z_w}), the weight of any fusion hypothesis $(I_1,\theta)$ is functionally related to the distance $d(\ell,\theta(\ell))$ between hypothetic objects in the association pairs included in $(I_1,\theta)$. Indeed, we have:
\begin{equation}\label{weight-GCI-divergence}
\begin{split}
w^{(I_1,\theta)}\propto&  \left(Q_1^{I_1}\right)^{\omega_1}\left(Q_2^{\theta(I_1)}\right)^{\omega_2}{\prod}_{\ell\in I_1}Z_\omega^{(\ell,\theta)}\\
=&\left(Q_1^{I_1}\right)^{\omega_1}\left(Q_2^{\theta(I_1)}\right)^{\omega_2}{\prod}_{\ell\in I_1}\exp^{-d(\ell,\theta(\ell))}.
\end{split}
\end{equation}
Hence, when fusion hypothesis $(I_1,\theta)$ includes a ``bad'' association pair $(\ell,\theta(\ell))$, then the corresponding weight $w^{(I_1,\theta)}$ after GCI-MB fusion becomes negligible, with no considerable contribution to the fusion results.
\subsection{Isolated clustering and the truncated GMB fused density}
In this subsection, we attempt to conveniently find and truncate all the negligible fusion hypotheses, which include at least one ``bad'' association pairs, by resorting to an isolated clustering of hypothetic objects of sensors $s=1,2$. In the following, a formal definition for a clustering is presented first. Then, we define an isolated clustering based on the concept of ``bad'' association pair as outlined in Definition~\ref{bad-association-pair}.
\begin{Def}\label{clustering}
A clustering $C=\{\mathcal{C}_1,\cdots,\mathcal{C}_{N_{\mathcal{C}}}\}$ of $\mathbb{L}_1$ and $\mathbb{L}_2$ is a set of clusters formed as paired subsets of hypothetic objects $\mathcal{C}_{g}=(\mathbb{L}_{1,g},\mathbb{L}_{2,g})$, where every cluster $g=1,\ldots,N_{\mathcal{C}}$ satisfies the following conditions:
\begin{itemize}
	\item $\mathbb{L}_{1,g} \subseteq \mathbb{L}_{1},$
	\item $\mathbb{L}_{2,g} \subseteq \mathbb{L}_{2},$
    \item $\mathbb{L}_{1,g}\cup\mathbb{L}_{2,g}\neq\emptyset.$
\end{itemize}
In addition, the subsets that pair to form the clusters are disjoint and partition the overall hypothetic objects, i.e.
\begin{itemize}
  \item $\mathbb{L}_1=\cup_{g=1}^{N_{\mathcal{C}}} \mathbb{L}_{1,g},\ \ \mathbb{L}_2=\cup_{g=1}^{N_{\mathcal{C}}} \mathbb{L}_{2,g},$
  \vspace{2mm}
  \item $\left((g,g')\in[1:N_{\mathcal{C}}]^2 ,  g\neq g'\right)\Rightarrow \left(\mathbb{L}_{1,g}\cap\mathbb{L}_{1,g'}\right)=\emptyset,$
  \vspace{2mm}
  \item $\left((g,g')\in[1:N_{\mathcal{C}}]^2 ,  g\neq g'\right)\Rightarrow \left(\mathbb{L}_{2,g}\cap\mathbb{L}_{2,g'}\right)=\emptyset.$
\end{itemize}
\end{Def}
Based on Definition \ref{clustering}, the hypothetic objects from different sensors belonging to different clusters are further referred to as \textit{inter-cluster association pairs}. Then, an isolated clustering is constructed in a principled way that any \textit{inter-cluster association pairs} of this clustering are ``bad'' association pairs. The formal definition is given as follows.
\begin{Def}\label{isolated-clustering}
A clustering $C=\{\mathcal{C}_1,\cdots,\mathcal{C}_{N_{\mathcal{C}}}\}$ for index sets $\mathbb{L}_1$ and $\mathbb{L}_2$ is said to be an isolated clustering, if $C$ satisfies:
\begin{equation}\label{clustering-criterion}
\min_{(\ell,\ell')\in \mathcal{P}} d(\ell,\ell')>\gamma,
\end{equation}
where
\begin{equation}\label{clustering-criterion2}
\mathcal{P}= \bigcup_{(g, g')\in[1:N_{\mathcal{C}}]^2, g\neq g'} \mathbb{L}_{1,g}\times\mathbb{L}_{2,g'}.
\end{equation}
\end{Def}
\begin{Rem}
In an isolated clustering $C$, any two different clusters $\mathcal{C}_g$ and $\mathcal{C}_{g'}$, $g\neq g'\in [1:N_{\mathcal{C}}]$ are said to be mutually isolated as they satisfy:
\begin{equation}
\min_{(\ell,\ell')\in\left(\mathbb{L}_{1,g}\times\mathbb{L}_{2,g'}\right)\bigcup \left(\mathbb{L}_{1,g'}\times\mathbb{L}_{2,g}\right)}d(\ell,\ell')>\gamma,
\end{equation}
where the union $\left(\mathbb{L}_{1,g}\times\mathbb{L}_{2,g'}\right)\cup \left(\mathbb{L}_{1,g'}\times\mathbb{L}_{2,g}\right)$ describes a set of inter-cluster association pairs for clusters $g$ and $g'$.  The isolation between cluster $g$ and $g'$ essentially demands that all the corresponding inter-cluster association pairs are ``bad'' association pairs.
\end{Rem}
A hypothesis $(I_1,\theta)$ is called an \emph{inter-cluster fusion hypothesis} if it includes at least one \emph{inter-cluster association pair},
\begin{equation}\label{inter-cluster-hypothesis}
(I_1,\theta)\bigcap \mathcal{P}\neq\emptyset.
\end{equation}
As a result, for an isolated clustering, due to the inclusion of inter-cluster association pair(s) (bad association pair(s)), all the \emph{inter-cluster fusion hypothesis} have a negligible contribution to the GMB density according to~(\ref{weight-GCI-divergence}), and they are exactly what we want to find and truncate. Discarding all the inter-cluster hypotheses, denoted by $\mathbb{D}=\{(I_1,\theta)\in\mathcal{F}(\mathbb{L}_1)\times\Theta: (I_1,\theta)\cap\mathcal{P}\neq\emptyset)\}$, the un-normalized GMB density given by (\ref{unnormalized}) can be approximated by
\begin{equation}\label{unnormalized_trun}
\begin{split}
&\widetilde\pi'_\omega(\{\bx_{1},\!\cdots\!,\bx_{n}\})\!=\\
&\!\sum_{\sigma}\!\!\sum_{(I_1,\theta)\in \mathcal{F}_n(\mathbb{L}_1) \times\Theta(I_1)-\mathbb{D}}\widetilde {w}_\omega^{(I_1,\theta)}\prod_{i=1}^{n}p_\omega^{([I_1]^v(i),\theta)}(\bx_{\sigma(i)}),
\end{split}
\end{equation}
and the normalization factor is given by
\begin{equation}\label{NF_trun}
\eta'_\omega=\sum_{(I_1,\theta)\in \mathcal{F}_n(\mathbb{L}_1) \times\Theta(I_1)-\mathbb{D}}\widetilde {w}_\omega^{(I_1,\theta)}.
\end{equation}
Hence, the normalized truncated GGI-MB density is
\begin{equation}\label{GMB-fused-density-discarded}
\pi'_\omega(X)=\widetilde\pi'_\omega(X)/\eta'_\omega.
\end{equation}
Omitting the inter-cluster hypotheses results in a heavily reduced number of hypotheses,
$N'_{H}\!=\!|\mathcal{F}(\mathbb{L})\times\Theta|\!-\!|\mathbb{D}|$. The truncated fused density (\ref{unnormalized_trun}) is also a transitional form of the independence approximation presented in the next subsection.
\subsection{Independence approximation of the fused GMB density}
Further than the truncated GMB fused density, in this subsection, we derive another equivalent form of equation (\ref{GMB-fused-density-discarded}), called the independence approximation, as outlined in Proposition 1.  The structure of this independence approximation suggests a parallelizable implementation of GCI fusions within individual clusters, in which the number of fusion hypotheses is even smaller than that of the truncated GMB fused density in most cases as concluded in Proposition 2.

Without loss of generality, an isolated clustering $C$ for index sets $\mathbb{L}_1$ and $\mathbb{L}_2$ is a union of three types of clusters,
\begin{equation}
C=C_{\text{I}}\cup C_{\text{II}}\cup C_{\text{III}},
\end{equation}
where
\begin{equation}\label{ccc}
\begin{split}
&\mathbb{L}_{1,g}\neq\emptyset\text{, }\mathbb{L}_{2,g}\neq\emptyset\text{, }\text{for  }  g\in C_{\text{I}} \\
&\mathbb{L}_{1,g}=\emptyset\text{, }\mathbb{L}_{2,g}\neq\emptyset\text{, }\text{for  }  g\in C_{\text{II}} \\
&\mathbb{L}_{1,g}\neq\emptyset\text{, }\mathbb{L}_{2,g}=\emptyset\text{, }\text{for  }  g\in C_{\text{III}}.
\end{split}
\end{equation}
For cluster types $C_{\text{II}}$ or $C_{\text{III}}$, the counterpart of one sensor in cluster $g$ is an empty set, because pairing any hypothetic object in that cluster with any hypothetic object at the other sensor yields a ``bad'' association pair.

In the following proposition, we show how a principled approximation to the original GMB density can be obtained by discarding all the insignificant inter-cluster hypotheses. This approximation not only consists of significantly smaller number of fusion hypotheses but also enjoys an appealing structure.
\begin{Pro}
Given an isolated clustering of index sets $\mathbb{L}_1$ and $\mathbb{L}_2$, $C=\{\mathcal{C}_1,\cdots,\mathcal{C}_{N_{\mathcal{C}}}\}=C_{\emph{I}}\cup C_{\emph{II}} \cup C_{\emph{III}}$, the normalized truncated GGI-MB density of form (\ref{GMB-fused-density-discarded}) can be expressed as
\begin{equation}\label{G-GCI-MB}
\pi'_\omega(X)
=\sum_{\mathop{\biguplus}\limits_{g:\mathcal{C}_g\in C_\emph{I} } X_g{= X}} \prod_{g: \mathcal{C}_g\in C_\emph{I}}\!\pi_{\omega,g}(X_g)
\end{equation}
where
the summation is taken over all mutually disjoint subsets $X_g,\,\, g\!:\!\mathcal{C}_g\in\!C_\emph{I}$ of $X$ such that $\bigcup_{g:\mathcal{C}_g\in C_\emph{I} } X_g = X$; each $\pi_{\omega,g}(\cdot)\,\,(g: \mathcal{C}_g\in C_\emph{I})$ is a GMB density  returned by the GCI fusion performed with the  MB distributions of the $g$th  cluster of Bernoulli components, i.e., $\pi_{s,g}=\{(r_s^{(\ell)},p_s^{(\ell)})\}_{\ell\in\mathbb{L}_{s,g}}, s=1,2$, i.e.,
\begin{equation}
\begin{split}
& \pi_{\omega,g}(\{\bx_1,\cdots,\bx_n\})\\=&\sum_{\sigma}
\sum_{(I_{1,g},\theta_g)\in \mathcal{F}_{n}\!(\mathbb{L}_{1,g}\!) \!\times\! \Theta_{g}({I_{1,g}})}\!\!\!\!\!\! w_\omega^{(I_{1,g},\theta_g)}\!{\prod}_{i=1}^{n}p_\omega^{([I_{1,g}]^v\!(i), \theta_g)}\!(x_{\sigma(i)}),
\end{split}
\end{equation}
where the injective function $\theta_g:\mathbb{L}_{1,g}\rightarrow\mathbb{L}_{2,g}$ denotes the fusion map of the $g$th cluster (without loss of generality, assume $\mathbb{L}_{1,g}\leq\mathbb{L}_{2,g} $), $\Theta_g$ denotes the set of all such fusion maps, and
\begin{align}
\label{fuse-w_p-g}w_{\omega,g}^{(I_{1,g},\theta_g)}&=\frac{\widetilde{w}_{\omega,g}^{(I_{1,g},\theta_g)}}{\eta_{\omega,g}}\\
\label{fuse-p_p-1}p_{\omega,g}^{(\ell,\theta_g)}(\bx)&=\frac{p_1^{(\ell)}(\bx)^{\omega_1}p_2^{(\theta_g(\ell))}(\bx)^{\omega_2}}{Z_{\omega,g}^{(\ell,\theta_g)}},\,\,\,\,\ell\in I_{1,g},
\end{align}
with
\begin{align}
\label{fuse-widetilde-w_p-g}\widetilde{w}_{\omega,g}^{(I_{1,g},\theta_g)}&=\left(Q_{1,g}^{I_{1,g}}\right)^{\omega_1}\left(Q_{2,g}^{\theta_g(I_{1,g})}\right)^{\omega_2}\!\!\prod_{\ell\in I_{1,g}}Z_{\omega,g}^{(\ell,\theta_g)}\\
\label{Z_w-g} Z_{\omega,g}^{(\ell,\theta_g)}&=\int p_1^{(\ell)}({x})^{\omega_1}p_2^{(\theta_g(\ell))}({x})^{\omega_2}d{x}\\
\label{K-1-g}\eta_{\omega,g}&=\sum_{I_{1,g}\in \mathcal{F}(\mathbb{L}_{1,g})}\sum_{\theta_g\in\Theta_g(I_{1,g})}\widetilde{w}_\omega^{(I_{1,g},\theta_g)}.
\end{align}
\end{Pro}

The detailed proof of Proposition 1 is given in Appendix A.  The expression (\ref{G-GCI-MB}) indicates that
\begin{itemize}
\item The fused GMB densities of different clusters in $C_{\text{I}}$ are mutually independent \cite{book_mahler} after discarding all the inter-cluster hypotheses (see [11, chap. 11, pg. 385]  for the standard expression for the independence of random finite subsets);
\item The clusters of types $C_{\text{II}}$ and $C_{\text{III}}$ disappear in this approximated density, since all the relevant hypotheses are inter-cluster hypotheses.
\end{itemize}
Thus, we refer to this approximated density as an \emph{independence approximation} since it can be factorized into several mutually independent and smaller size multi-object densities.

The structure of this independence approximation suggests an efficient and intuitive implementation of GCI-MB fusion, namely parallelizable implementation of GCI fusions within individual clusters. In such an implementation, the total number of fusion hypotheses over all clusters becomes
\begin{equation}
N''_{H}=\sum_{g: \mathcal{C}_g\in C_{\text{I}}}\left(|\mathcal{F}(\mathbb{L}_{1,g})\times\Theta_{g}|\right).
\end{equation}
\begin{Pro}
As long as $|\mathbb{L}_{1,g}|\geq 1$ with ${g: \mathcal{C}_g\in C_\emph{I}}$, and {{$N_{\mathcal{C}_{\emph{I}}}=|C_{\emph{I}}|\geq 2$,}} the following holds:
\begin{equation}N''_H\leq N'_H,\end{equation}
where the equality holds if and only if $|\mathbb{L}_{1,g}|= 1$, $|\mathbb{L}_{2,g}|=1$ and $N_{\mathcal{C}_\emph{I}}= 2$.
\end{Pro}
The proof of Proposition 2 is given in Appendix B. It shows that in most cases, the number of fusion hypotheses over all clusters is even smaller than that of only subtracting the inter-cluster hypotheses, which is the case of the truncated fused GMB density given in the previous subsection. Moreover, due to its parallelizable structure, the execution time only depends on the cluster with the largest number of hypothetic objects.

\subsection{Approximation error}
The only source of error in the independence approximation~(\ref{G-GCI-MB}) is due to the omitting of the inter-cluster hypotheses. The following proposition establishes the upper bound of the $L_1$ error between (\ref{G-GCI-MB}) and the original GMB density in (\ref{GCI-MB}). The $L_1$ error was used by Vo~et~al.~\cite{LMB_Vo2} to analyze the approximation error between the GLMB density and its truncated version.
\begin{Pro}\label{Pro_error}
Let $\|f\|_1\triangleq\int |f(X)|\delta X$ denotes the $L_1$-norm of $f:\mathcal{F}(\mathbb{X})\rightarrow\mathbb{R}$. The following results hold:

\noindent(1) the $L_1$-error between $\pi_{\omega}(\cdot)$ of (\ref{GCI-MB}) and ${\pi}'_{\omega}(\cdot)$ of (\ref{G-GCI-MB}) satisfies,
\begin{equation}
\begin{split}\label{errorBound}
 \left\| \pi_{\omega}(\cdot)-\pi'_{\omega}(\cdot)\right\|_1&\!\leqslant 
 2\sum_{(I_1,\theta)\in\mathbb{D}}w_\omega^{(I_1,\theta)}\\
 \!&\leqslant A \exp\left(\!-\!\gamma\right)
\end{split}
\end{equation}
where $\gamma$ is the GCI divergence threshold in (\ref{GCI-distance-threshold}) to define an bad association pair, and
\begin{align}
A=&2\!\!\sum_{(I_1,\theta)\in\mathbb{D}}\!\!K^{(I_1,\theta)}\big/ \eta_\omega\\
K^{(I_1,\theta)}=&\left(Q_1^{I_1}\right)^{\omega_1}\left(Q_2^{\theta(I_1)}\right)^{\omega_2};
\end{align}

\noindent(2) if the clustering threshold $\gamma\rightarrow\!+\infty$,   $\left\| \pi_{\omega}(\cdot)\!-\!{\pi'}_{\omega}(\cdot)\right\|_1\rightarrow\! 0$.
\end{Pro}

See Appendix C for the proof of Proposition 3. It follows from the above results that the upper bound of $L_1$-error between (\ref{GCI-MB}) and (\ref{G-GCI-MB}) is determined by the weights of inter-cluster hypotheses in $\mathbb{D}$. This further verifies that the only source of error in the approximation given by equation~(\ref{G-GCI-MB}) is the truncation of inter-cluster hypotheses. In addition, it can be seen from (\ref{errorBound}) that the upper bound of the $L_1$-error decreases exponentially with the clustering threshold.
When the clustering threshold $\gamma\!\rightarrow\!+\infty$, the $L_1$-error will be zero, which indicates that a sufficiently large clustering threshold  can guarantee a high accuracy of the independence approximation.

\begin{Rem}
The $L_1$-error between (\ref{GCI-MB}) and (\ref{G-GCI-MB}) can also be viewed as a measure of statistical dependence between the factorized GMB densities of different clusters. The smaller this $L_1$-error is, the weaker the dependence between these factorized densities will be. When this $L_1$-error approaches zero, the factorized densities of different clusters become closer to being mutually independent~\cite{book_mahler}, and the factorization~(\ref{G-GCI-MB}) approaches the exact density.
\end{Rem}

\subsection{Disjoint-set data structure-based fast clustering}
As described in Section IV-B, by carefully clustering the Bernoulli components into isolated clusters according to GCI divergence, and discarding the inter-cluster hypotheses, we can obtain an efficient approximation of the GMB density. There are two challenges in the designing of the clustering routine. The first challenge is that clustering needs to be accomplished beforehand without exhausting all the fusion hypotheses. The second challenge is to seek an isolated clustering with the largest number of clusters, referred to as the \textit{Largest Isolated Clustering} (LIC) as defined below.

\begin{Def}
	Consider a clustering of index sets $\mathbb{L}_1$ and $\mathbb{L}_2$, $C=\{\mathcal{C}_1,\cdots,\mathcal{C}_{N_{\mathcal{C}}}\}$. For a given cluster $\mathcal{C}_g$, $g\in [1:N_{\mathcal{C}}]$, consider two subsets
\begin{equation}\mathbb{L}_{1,g}^a\subseteq\mathbb{L}_{1,g},\ \  \mathbb{L}_{2,g}^a\subseteq\mathbb{L}_{2,g}
\end{equation}
 and define
\begin{equation}
	\mathbb{L}_{1,g}^b \triangleq \mathbb{L}_{1,g} \backslash \mathbb{L}_{1,g}^a,\ \
	\mathbb{L}_{2,g}^b \triangleq \mathbb{L}_{2,g} \backslash \mathbb{L}_{2,g}^a.
\end{equation}
	$C$ is an indivisible clustering if for any cluster $\mathcal{C}_g$ and any subsets $\mathbb{L}_{1,g}^a$ and $\mathbb{L}_{2,g}^a$ that are not both empty, the cluster $(\mathbb{L}^a_{1,g}, \mathbb{L}^a_{2,g})$ and the cluster $(\mathbb{L}^b_{1,g}, \mathbb{L}^b_{2,g})$ are not mutually isolated.
\end{Def}

\begin{Def}\label{LIG}
A clustering of index sets $\mathbb{L}_1$ and $\mathbb{L}_2$,  $C=\{\mathcal{C}_1,\cdots,\mathcal{C}_{N_{\mathcal{C}}}\}$  is the LIC,  if it is both an isolated clustering, and indivisible.
\end{Def}
\begin{Rem}
In an indivisible clustering $C$,  each cluster $\mathcal{C}_g$ can not be divided into smaller size isolated clusters any more, which guarantees the largest number of isolated clusters.
\end{Rem}
Note that based on Definitions 4 and 5, the isolated clusterings are not unique, but for any pair of index sets $\mathbb{L}_1$ and $\mathbb{L}_2$, there is only one LIC. Truncation of the inter-cluster hypotheses of the LIC guarantees that all the insignificant GCI-MB fusion hypotheses are discarded.

In the following, by modeling the underline data structure of hypothetic objects among all sensors as a undirected graph, we present how the formation of the LIC is tantamount to seeking connected components of an undirected graph dynamically, which can be solved by the disjoint-set data structure-based fast solver with a computational expense that is only polynomial in the number of hypothetic objects. In addition, using this solver, the search for connected components can be performed at an early stage, since it only takes the pairwise distances between hypothetic objects of two sensors as inputs.

\subsubsection{Structure modeling of hypothetic objects}
This subsection presents the construction of  the undirected graph  using the hypothetic objects of sensors $s=1,2$, i.e., $(\mathbb{L}_1,\mathbb{L}_2)$, and their mutual relationship. Specifically, the constructions of vertices, edges and the paths of the undirected graph are as follows:
\begin{itemize}
	\item[$\bullet$] \textit{The vertex set} $V$: Each hypothetic object of sensor 1, $\ell\in\mathbb{L}_1$ is considered as a vertex, and then all the hypothetic objects of sensor 1 form the vertex set, i.e., $V=\mathbb{L}_1$.

\item[$\bullet$] \textit{The edge set} $E$: To define the relationship between each pair of vertices, firstly, for each hypothetic object $\ell\in\mathbb{L}_1$, extract all associated hypothetic objects from $\mathbb{L}_2$ which fall within the gate, i.e.,
\begin{equation}
\begin{split}\label{cluster_1}
\Psi^{(\ell)}_{2}=\{\ell'\in\mathbb{L}_2: d(\ell,\ell')\leq\gamma\}.
\end{split}
\end{equation}
If any two hypothetic objects $\ell_1\neq \ell_2\in\mathbb{L}_1$ have common associated hypothetic objects from sensor 2, i.e.,
\begin{equation}
\begin{split}\label{edge}
\Psi^{(\ell_1)}_{2}\cap\Psi^{(\ell_2)}_{2}\neq\emptyset,
\end{split}
\end{equation}
then $\ell_1$ and $\ell_2$ are paired to be an edge,  $e=(\ell_1,\ell_2)$. All pairs of hypothetic objects in $\mathbb{L}_1$ satisfying (\ref{edge}) form the edge set $E$,
\begin{equation}
E=\{(\ell_1,\ell_2)\in\mathbb{L}_1^2: \ell_1\neq\ell_2\ \ , \Psi_{2}^{(\ell_1)}\cap\Psi_{2}^{(\ell_2)}\neq\emptyset\}.
\end{equation}

\item[$\bullet$] \textit{A path in} $G(V,E)$ is a walk in which all vertices (except possibly the first and last) are distinct, and all edges are distinct. A path of length $K$ in a graph is an alternating sequence of vertices and edges, i.e., $\ell_0, e_0,\ell_1,e_1\cdots,e_{K-1},\ell_{K}$, where $e_{k-1}$ connects vertices $\ell_{k-1}$ and $\ell_k$, $k=1,\cdots,K$, $\ell_0$ and $\ell_{K}$ are the end-points, and $\ell_1,\cdots,\ell_{K-1}$ are the non-end-points.
\end{itemize}
Fig. \ref{Model-undirected-graph_2} shows a sketch map for the construction of the undirected graph  using the hypothetic objects of sensors 1 and  2.
\begin{figure}[!h]
  \centering
\includegraphics[width=8cm]{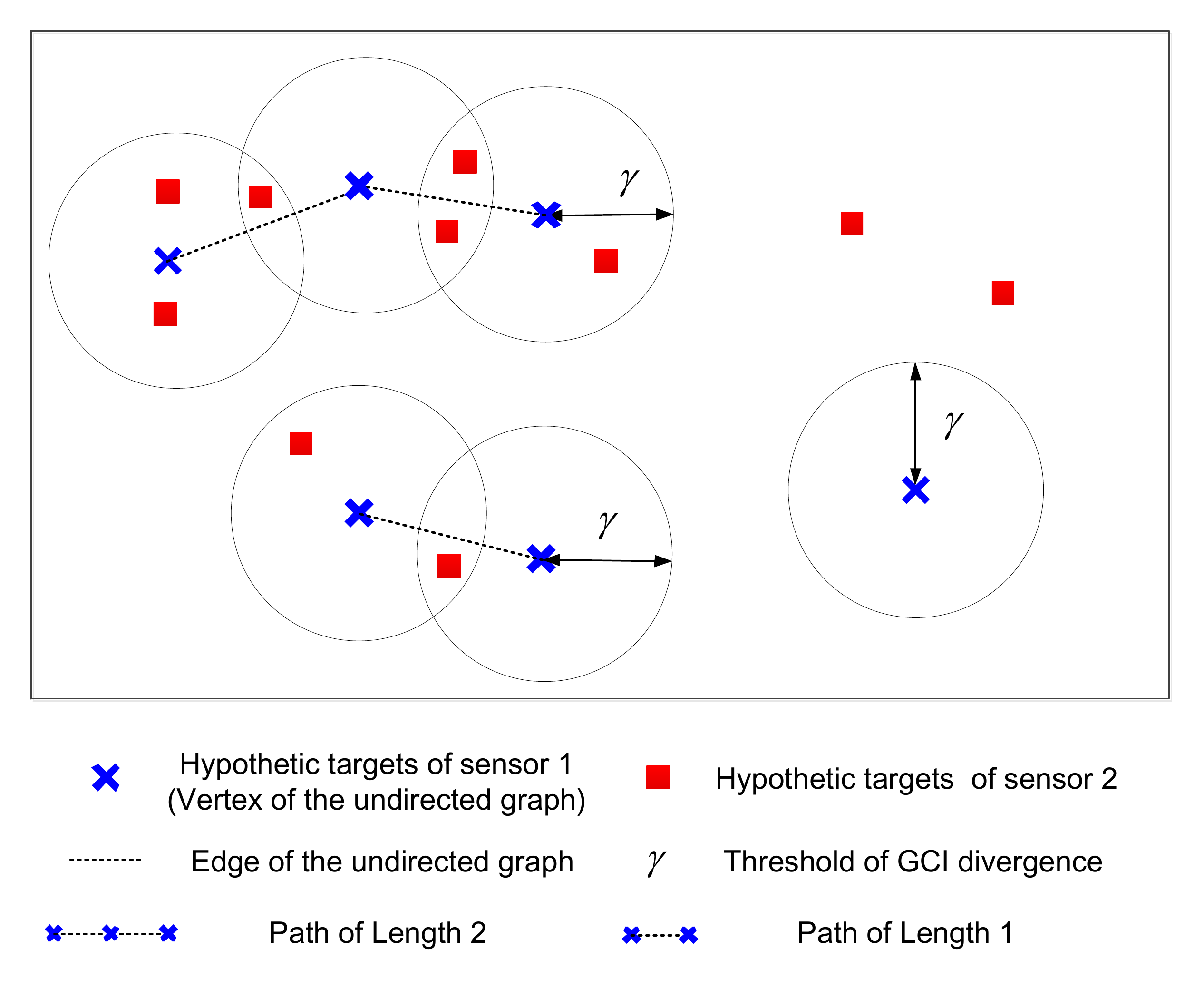}
\caption{A sketch map for the construction of the undirected graph  using the hypothetic objects of sensors 1 and  2.\label{Model-undirected-graph_2}}
\end{figure}
\subsubsection{Clustering based on union-find algorithm}
This subsection details how to get the LIC by utilizing the undirected graph $G(V,E)$. Recall (\ref{ccc}), $C_{\text{III}}$ can be constructed as follow:
\begin{equation}\label{type3}
C_{\text{III}}=\{(\emptyset, \{\ell'\}):\ell'\in \mathbb{L}_2\backslash (\cup_{\ell\in\mathbb{L}_1}\Psi_2^{(\ell)})\}.\end{equation}
The counterpart of sensor 2 in any cluster belonging to $C_{\text{III}}$ is a singleton whose element is an individual hypothetic object of sensor 2 having no associated hypothetic objects from $\mathbb{L}_1$. The choice of the singleton is to ensure the largest number of clusters.

Then, construction of $C_{\text{I}}$ and $C_{\text{II}}$ is the problem of seeking connected components of the undirected graph $G(V,E)$. A connected component of an undirected graph is a subgraph that satisfies two conditions:

-- any two vertices are connected to each other by a path;

-- it is connected to no additional vertices in the supergraph.

\begin{algorithm}[bt]\label{uion-find}
\caption{Union\_Find\_Algorithm.}
\textbf{Input: } The undirected graph $G(V,E)$\;
\For {each $\ell\in V$}
{\textsc{Makeset}$(\ell)$}  
\For {each $(\ell_1,\ell_2)\in E$}
{
\If{\textsc{Find}$(\ell_1)\neq$\textsc{Find}$(\ell_2)$  }
{\textsc{Union}$(\ell_1,\ell_2)$}}
\textbf{Output: } Connected components of $G$.
\end{algorithm}

Dynamic tracking of the connected components of a graph as vertices and edges is a straightforward application of disjoint-set data structures {{\cite{Unionfind3}}}. A disjoint-set data structure maintains a dynamic collection of disjoint subsets that together cover the whole set. The underline data structure of each set is typically a rooted-tree, identified by a representative element of the set, also called the parent node. The two main operations are to \emph{\textsc{Find}} which set a given element belongs to by locating its father node and to replace two existing sets with their \textit{\textsc{Union}}.  We outline a primary \textsc{Union-Find} algorithm from the viewpoint of its use in finding connected components of the graph $G(V,E)$ as shown in {Algorithm 1}.

The procedure of seeking connected components can be expanded as follows. Firstly, initialize $|V|$ trees, each of which takes an individual element of $V$ as the parent node and a rank of 0, which is exactly the operation \textsc{Makeset}. Then, for each edge $(\ell_1,\ell_2)\in E$, find their parent nodes through the operation \textsc{Find}. If parent nodes of $\ell_1$ and $\ell_2$ are the same, i.e. if they have already been in the same tree, then no operation is needed; otherwise, merge the trees which $\ell_1$ and $\ell_2$ belong to, respectively, using the operation \textsc{Union}. Finally, we get all the connected components of {$G(V,E)$}, denoted by $G_g(V_g,E_g),g=1,\cdots,N_G$.

Accordingly, $\mathbb{L}_1$ is partitioned into $N_G$ disjoint subsets, i.e., \begin{equation}\mathbb{L}_{1,g}=V_g, g=1,\cdots,N_G;
\end{equation}
and for each $\mathbb{L}_{1,g}$, all the associated hypothetic objects from sensor 2 are merged, i.e.,
\begin{equation}\label{union-find-clustering}\mathbb{L}_{2,g}=\cup_{\ell\in \mathbb{L}_{1,g}} \Psi_{2}^{(\ell)}.
\end{equation}
Then each pair of $\mathbb{L}_{1,g}$ and $\mathbb{L}_{2,g}$ forms an isolated cluster $\mathcal{C}_g=(\mathbb{L}_{1,g},\mathbb{L}_{2,g})$, and all the isolated clusters form the clusters of types $C_{\text{I}}\cup C_\text{II}$ in the finalized LIC,
\begin{equation}\label{type1-2}
C_{\text{I}}\cup C_\text{II}=\{\mathcal{C}_1,\cdots,\mathcal{C}_{N_G}\}.
\end{equation}
\begin{Rem} Any vertex satisfying $\Psi_2^{(\ell)}=\emptyset$ is a connected component, since it cannot be linked to any other vertex. Hence, it is easy to separate $C_{\emph{II}}$ from the union in (\ref{type1-2}), i.e.,
\begin{equation}\label{type2}
C_{\emph{II}}=\{(\{\ell\},\emptyset):\Psi_2^{(\ell)}=\emptyset,\ell\in\mathbb{L}_1\},
\end{equation}
for which the counterpart of sensor 1 in any cluster is a singleton whose element is an individual hypothetic object having no associated hypothetic object from sensor 2.
\end{Rem}
\begin{Pro}\label{Pro_cluster}
The union $C=C_\emph{I}\cup C_\emph{II} \cup C_\emph{III}$ with $C_\emph{I}$, $C_\emph{II}$ and $C_\emph{III}$, respectively,  given in (\ref{type1-2}) and (\ref{type3}) is the LIC.
\end{Pro}
The proof of Proposition~\ref{Pro_cluster} is given in Appendix D. Proposition~\ref{Pro_cluster} ensures that one can obtain a unique LIC using the disjoint-set data structure-based fast clustering, thus building exactly the independence approximation of form (\ref{G-GCI-MB}).

\begin{Rem}
Algorithm 1 only shows the primary \textsc{Union-Find} algorithm. Actually, there exist many classical theoretical studies  providing enhanced implementations of the \textsc{Union-Find} algorithm, including Link-by-Rank or Link-by-Size for the union operation and Path-Compression, Path-Splitting, or Path-Halving for the find operation~\cite{unionfind_3,unionfind_2}. The worst-case time complexity of the enhanced union-find algorithm is $\mathcal{O}(n+m\alpha(m,n))$ for any combination of $m$ \textsc{Makeset}, \textsc{Union}, and \textsc{Find} operations on $n$ elements, where $\alpha$ is the very slowly growing
inverse of Ackermann¡¯s function.
\end{Rem}

\subsection{Summary}
The combination of the independence approximation~(\ref{G-GCI-MB}) and Proposition~\ref{Pro_cluster} enables a computationally efficient GCI-MB fusion algorithm. The fast clustering algorithm is employed at an early stage to form the LIC and thus obtain the approximated density. Then GCI-MB fusion is performed within each smaller clusters of LIC independently and parallelizablly by exploiting the structure of the approximated density. We refer to this fusion method as the parallelizable GCI-MB (P-GCI-MB) to emphasize its intuitive implementation structure.

Note that Proposition \ref{Pro_cluster} also ensures that the only approximation error arises from discarding the inter-cluster hypotheses. The analysis in Proposition \ref{Pro_error} further shows that as long as the clustering threshold is sufficiently large, the approximation error is negligible. Hence, the proposed P-GCI-MB method is able to reduce the computational expense as well as memory requirements with slight and bounded error.
\begin{Rem}
Theoretically in the worst case where all hypothetic objects are in close proximity, it may not be possible to partition objects into smaller clusters. In this case, the complexity of P-GCI-MB is the same as that of the original GCI-MB fusion.
\end{Rem}

\section{Implementation of P-GCI-MB Fusion}
In this section, two common implementation approaches of the P-GCI-MB Fusion,  i.e., GM and SMC approaches are presented. The GM approach is suitable for the case of linear or weak non-linear models, while the SMC approach is general enough for both linear and non-linear models.  Compared with SMC approach, the advantage of GM approach lies on its lower load in terms of both local (in-node) computations and inter-node data
communications.



\subsection{The GM implementation}
At the first step of the P-GCI-MB fusion, the quantities needed to calculate are the distances (i.e., GCI divergences)  $d(\ell,\ell'), (\ell,\ell')\in\mathbb{L}_1\times\mathbb{L}_2$ in (\ref{GCI-distance}). Then, we need to calculate $p_\omega^{(\ell,\theta)}$ and $w^{(I_{1},\theta)}$ for each fusion hypothesis $(I_{1},\theta)$ of each cluster. Since their calculations involve the integrations which do not have closed form solutions in general, we resort to GM realization. Having the results of quantities $d(\ell,\ell')$, $p_\omega^{(\ell,\theta)}$ and $w^{(I_{1},\theta)}$, other terms required in P-GCI-MB can be calculated accordingly.
\subsubsection{GM evaluation of $d(\ell,\ell')$}
The pairwise distances between hypothetic tracks of two sensors need to be computed first since they are the inputs to the clustering algorithm in order to obtain the LIC. Using the GM approximation, the posterior of the $s$-th local sensor can be expressed as $\{(r_{s}^{(\ell)},p_s^{(\ell)}(\bx))\}_{\ell\in\mathbb{L}_s}$, with each probability density $p_s^{(\ell)}(\bx)$ being represented as a mixture of Gaussian components:
 \begin{equation}
\begin{split}\label{GMs}
p_{s}^{(\ell)}(\bx)={\sum}_{j=1}^{J_{s}^{(\ell)}}\alpha_{s}^{(\ell,j)}\mathcal{N}\left(\bx;m_{s}^{(\ell,j)},P_{s}^{(\ell,j)}\right)
\end{split}
\end{equation}
where $\mathcal{N}\left(\bx;m,P\right)$ denotes a Gaussian density with mean $m$ and covariance $P$, $J_{s}^{(\ell)}$ is the number of Gaussian components for the $\ell$-th Bernoulli component, $\alpha_{s}^{(\ell,j)}$ is the weight of the $j$-th Gaussian component for the $\ell$-th Bernoulli component.

Observing (\ref{p_l}) and (\ref{GCI-distance}), we find that calculation of both $p_\omega^{(\ell,\theta)}$ and $d(\ell,\ell')$ involves integrating the Gaussian mixture in~(\ref{GMs}) taken to the power of $\omega_s$, which in general has no analytical solution. In this respect, we adopt the principled approximation suggested by~Battistelli~et~al.\cite{Battistelli} and Julier~\cite{July},
 \begin{equation}\label{app_gm}
\begin{split}
&\left[{\sum}_{j=1}^{J_{s}^{(\ell)}}\alpha_{s}^{(\ell,j)}\mathcal{N}\left(\bx;m_{s}^{(\ell,j)},P_{s}^{(\ell,j)}\right)\right]
^{\omega_s}\\
&\hspace{17mm}\approxeq {\sum}_{j=1}^{J_{s}^{(\ell)}}\left[\alpha_{s}^{(\ell,j)}\mathcal{N}\left(\bx;m_{s}^{(\ell,j)},P_{s}^{(\ell,j)}\right)\right]^{\omega_s}
\end{split}
\end{equation}
with the condition that the Gaussian components of $p_{s}^{(\ell)}(x)$ are well separated relatively to their corresponding covariances. If this condition is not satisfied, one can either perform merging before fusion (this is possible, since, for an MB, each probability  density $p_s^{(\ell)}(\cdot)$ is corresponding to a  single hypothetic object) or use other approximations, e.g., replacing the GM representation by a sigma-point approximation \cite{fractional-power}.

Using (\ref{app_gm}), the term $p_1^{(\ell)}(\bx)^{\omega_1}p_2^{(\ell')}(\bx)^{\omega_2}$ is expressed as
  \begin{equation}
\begin{split}\label{density_Bernoulli}
&p_1^{(\ell)}(\bx)^{\omega_1} p_2^{(\ell')}(\bx)^{\omega_2} \\
 & \hspace{6mm}= \sum_{j=1}^{J_{1}^{(\ell)}}\sum_{j'=1}^{J_{2}^{(\ell')}}\alpha_{\omega}^{(\ell,j,\ell',j')}\mathcal{N}\left(\bx;m_{12}^{(\ell,j,
\ell',j')},P_{12}^{(\ell,j,\ell',j')}\right)
\end{split}
\end{equation}
where
\begin{align}
\alpha_{\omega}^{(\ell,j,\ell',j')}\!=&\widetilde{\alpha}_{\omega}^{(\ell,j,\ell',j')}
\mathcal{N}\!\!\left(\!\!m_{1}^{(\ell,j)}\!\!-\!m_{2}^{(\ell',j')};0,\!\frac{P_{1}^{(\ell,j)}}{\omega_1}\!\!+\!\!\frac{P_{2}^{(\ell',j')}}{\omega_2}\!\!\right)\\
P_{12}^{(\ell,j,\ell',j')}\!=\!&\left([P_{1}^{(\ell,j)}]^{-1}+[P_{2}^{(\theta(\ell),j')}]^{-1}\right)^{-1}\\
m_{12}^{(\ell,j,\ell',j')}\!=&P_{12}^{(\ell,j)}\left([P_{1}^{(\ell,j)}]^{-\!1}m_{1}^{(\ell,j)} \!\!+\![P_{2}^{(\ell',j')}]^{-\!1}m_{2}^{(\ell',j')} \right)
\end{align}
with
\begin{align}
\widetilde{\alpha}_{\omega}^{(\ell,j,\ell',j')}&=(\!\alpha_{1}^{(\ell,j)}\!)^{\omega_1}(\!\alpha_{2}^{(\ell',j')}\!)^{\omega_2}\!\rho(P_{1}^{(\ell,j)}\!, \!\omega_1)\!\rho(P_{2}^{(\ell',j')}\!, \!\omega_2)\\
 \rho(P,\omega)&=\sqrt{\det[2\pi P\omega^{-1}](\det[2\pi P])^{-\omega}}.
\end{align}
Substituting (\ref{density_Bernoulli}) into (\ref{GCI-distance}) and performing integration, we have
\begin{equation}\label{GCI-distance-GM}
d(\ell,\ell')=-\log \left({\sum}_{j=1}^{J_{1}^{(\ell)}}{\sum}_{j'=1}^{J_{2}^{(\ell')}}\alpha_{\omega}^{(\ell,j,\ell,j')}\right).
\end{equation}

\subsubsection{Evaluation of $p_\omega^{(\ell,\theta)}(\cdot)$ and $w_\omega^{(I_1,\theta)}$} After discovering the LIC based on the pairwise distances~(\ref{GCI-distance-GM}), for each cluster $g$, let $\ell\in\mathbb{L}_{1,g}$, $\theta\in\Theta_g$. Substituting  $\ell'=\theta(\ell)$ into (\ref{density_Bernoulli}),  we can obtain the numerator of density $p_\omega^{(\ell,\theta)}(\cdot)$ as
  \begin{equation}
\begin{split}\label{density_Bernoulli_theta}
&p_1^{(\ell)}(\bx)^{\omega_1} p_2^{(\theta(\ell))}(\bx)^{\omega_2} \\
=&  \sum_{j=1}^{J_{1}^{(\ell)}}\sum_{j'=1}^{J_{2}^{(\theta(\ell))}}\alpha_{\omega}^{(\ell,j,\theta(\ell),j')}\mathcal{N}\left(\bx;m_{12}^{(\ell,j,\theta(\ell),j')},\!P_{12}^{(\ell,j,\theta(\ell),j')}\right).
\end{split}
\end{equation}
Further, the parameter $Z_\omega^{(\ell,\theta)}$ in (\ref{Z_w}) can be computed by
  \begin{equation}
\begin{split}\label{integrate-gm}
Z_\omega^{(\ell,\theta)}=\exp(-d(\ell,\theta(\ell)).
\end{split}
\end{equation}
By substituting (\ref{integrate-gm}) and (\ref{density_Bernoulli_theta}) into (\ref{fuse-p_p}), the parameter $p_\omega^{(\ell,\theta(\ell))}$ turns out to be
  \begin{equation}
\begin{split}\label{fuse-p-gm}
&p_{\omega}^{(\ell,\theta(\ell))}(\bx)=\frac{p_1^{(\ell)}(\bx)^{\omega_1} p_2^{(\theta(\ell))}(\bx)^{\omega_2}}{Z_\omega^{(\ell,\theta)}}\\
=&\frac{\sum_{j=1}^{J_{1}^{(\ell)}}\sum_{j'=1}^{J_{2}^{(\theta(\ell))}}\alpha_{\omega}^{(\ell,j,\theta(\ell),j')}\mathcal{N}\left(\bx;m_{12}^{(\ell,j,\theta(\ell),j')},P_{12}^{(\ell,j,\theta(\ell),j')}\right)}
 {\exp(-d(\ell,\theta(\ell))}.
\end{split}
\end{equation}
By combing (\ref{integrate-gm}) with (\ref{fuse-w_p}), (\ref{K-1}), the un-normalized weight $\widetilde{w}_\omega^{(I_1,\theta)}$ for each $(I_1,\theta)\in\mathcal{F}(\mathbb{L}_{1,g})\times\Theta_{g}(I_1)$ and the normalization constant $\eta_\omega$ are calculated as
\begin{align}
\label{fuse-w_p-gm}
\widetilde{w}_\omega^{(I_1,\theta)}&=\left(Q_1^{I_1}\right)^{\omega_1}\left(Q_2^{\theta(I_1)}\right)^{\omega_2}\prod_{\ell\in I_1}\exp(-d(\ell,\theta(\ell))\\
\label{K-1-gm}\!\!\eta_\omega&={\sum}_{I_1\in \mathcal{F}(\mathbb{L}_1)}{\sum}_{\theta\in\Theta_g(I_1)}\widetilde{w}_\omega^{(I_1,\theta)}.
\end{align}
Thus, by substituting (\ref{fuse-w_p-gm}) and (\ref{K-1-gm}) into (\ref{fuse-w}), the GM representation of the normalized weight $w_\omega^{(I_1,\theta)}$ is calculated.
\subsection{SMC implementation}
The SMC realization of GCI-MB was previously given in \cite{GCI-MB}. Since the P-GCI-MB fusion contains parallelizable GCI-MB fusion operations with several smaller clusters, the SMC implementation of P-GCI-MB fusion for each cluster is straightforward and similar to \cite{GCI-MB}. As for the pair-wise distances $d(\ell,\ell'), (\ell,\ell')\in\mathbb{L}_{1}\times\mathbb{L}_{2}$ which are the outputs of the clustering algorithm,  one  can also refer to the SMC evaluation of $Z_\omega^{(\ell,\theta)}$ given in \cite{GCI-MB},  since the distance $d(\ell,\theta(\ell))$ is functionally related to $Z_{\omega}^{(\ell,\theta)}$ following (\ref{integrate-gm}).
%
\subsection{Pseudo-code of P-GCI-MB fusion }
The psedo-code of P-GCI-MB fusion are given in Algorithms 2 and 3. Specifically, {Algorithm 2} forms the LIC using the disjoint-set data structure; and {Algorithm 3} describes the whole fusion algorithm in a two-sensor case. This implementation can also be easily extended to $N_s\geqslant 2$ sensors by sequentially applying the pairwise fusion (\ref{r_l}) and (\ref{p_l}) under each clusters $N_s-1$ times, where the ordering of pairwise fusions is irrelevant. Similar mechanism has been used in CPHD filters and LMB filters based GCI fusions \cite{Battistelli, LMB-Reuter}. In addition, by exploiting the consensus approach, the P-GCI-MB fusion can be also implemented in a fully distributed way, similar to \cite{Battistelli}. The consensus based P-GCI-MB fusion achieves the global P-GCI-MB fusion over  the whole  network by iterating the local fusion steps among neighbouring sensor nodes.
\begin{algorithm}[hbt]\label{algorithm: GCI-MB Fusion Table}
\caption{Pseudo-code of the Clustering\_Function.}
\begin{minipage}{0.93\columnwidth}
\underline{Input: } $\pi_s=\!\{(r_s^{(\ell)},p_s^{(\ell)}(\cdot))\}_{\ell\in\mathbb{L}_s}$ from nodes $s=1,2$\\
\underline{Output:} The LIC ${C}$ and the distance matrix $D_{\omega}$ \\
\textbf{function} Clustering\_Function\,$(\pi_1,\pi_2)$
\\
\textsc{Define} a undirected graph, $G=(V,E)$ with $V=\mathbb{L}_1$\;
\textsc{Initialize} the set of edges, $E=\emptyset$\;
\textsc{Initialize}  a distance matrix $D_\omega=\emptyset$ with $|\mathbb{L}_1|$ columns and $|\mathbb{L}_2|$ rows\;
\textsc{Initialize} the row number of $D_\omega$, $i=0$\;
\textsc{Initialize} the column number of  $D_\omega$, $j=0$\;
\For {$\ell=1:|\mathbb{L}_1|$}
{$\Psi_{2}^{(\ell)}=\emptyset$\;
 \For {$\ell'=1:|\mathbb{L}_2|$}
{Evaluate the distance parameter $d(\ell,\ell')$; \hfill{$\rhd$(\ref{GCI-distance})}\;
\If {$d(\ell,\ell')<\Gamma$}{$\Psi_{2}^{(\ell)}=\Psi_{2}^{(\ell)}\cup\{\ell'\}$ \;
} $D_\omega[i+1,j+1]=d(\ell,\ell')$\;
$i=:i+1$}
\If{$\ell>1$}{
\For {$\ell_{\text{temp}}=1:\ell$}
{\If{$\Psi_{2}^{{(\ell)}}\cap\Psi_{2}^{(\ell_{\text{temp}})}\neq\emptyset$}
{$E=E\cup\{(\ell,\ell_{\text{temp}})\}$}}}
$j=:j+1$}
$[\mathbb{L}_{1,1},\cdots,\mathbb{L}_{1,N_\mathcal{C}}]=\text{Union\_Find\_Algorithm}(G(V,E))$\;
\For {$g=1:N_{\mathcal{C}}$}
{$\mathbb{L}_{2,g}=\cup_{\ell\in\mathbb{L}_{1,g}} \Psi_{2}^{(\ell)}$\;
}
\textbf{Return:} The LIC ${C}=\{(\mathbb{L}_{1,g}, \mathbb{L}_{2,g})\}_{g=1}^{N_\mathcal{C}}$ and  $D_{\omega}$.
\end{minipage}
\end{algorithm}

\begin{algorithm}[tb]\label{algorithm: GCI-MB-2}
\caption{Pseudo-code of the GM P-GCI-MB fusion}
\begin{minipage}{0.93\columnwidth}
\underline{Input: } $\pi_s=\!\{(r_s^{(\ell)},p_s^{(\ell)}(\cdot))\}_{\ell\in\mathbb{L}_s}$ from nodes $s=1,2$\\
\underline{Output:} The fused density  $\pi_\omega=\!\{(r_\omega^{(\ell)},p_\omega^{(\ell)}(\cdot))\}_{\ell\in\mathbb{L}_\omega}$\\
\textbf{function:} P\_GCI\_MB\_Fusion$(\pi_1,\pi_2)$\\
${C}=$\,Clustering\_Function\,$(\pi_1,\pi_2)$\;
\textsc{Initialize} the fused MB parameter set, $\pi_\omega=\emptyset$\;
\For{$g=1:N_\mathcal{C}$}
{\textsc{Create} the map space $\Theta_{g}=\{\theta_g:\mathbb{L}_{1,g}\rightarrow\mathbb{L}_{2,g}$\}\;
\For {$I_1\in\mathcal{F}(\mathbb{L}_{1,g})$}
{
\For{$\theta\in\Theta_g(I_1)$}
{
\For{$\ell \in I_1$}
{
          \textsc{Find} the quantity $d_\omega^{(\ell,\theta(\ell))}$ in  $D_\omega$\;
          \textsc{Calculate} the quantity $Z_\omega^{(\ell,\theta)}$;  \hfill{$\rhd$(\ref{integrate-gm})}
          \textsc{Evaluate} the quantity $p_{\omega}^{(\ell,\theta)}(\bx)$; \hfill{$\rhd$(\ref{fuse-p-gm})}\\
}
\textsc{Calculate} the weight $\widetilde w_\omega^{(I_1,\theta)}$; \hfill{$\rhd$(\ref{fuse-w_p-gm})}}
}
\textsc{Calculate} normalized factor $\eta$; \hfill{$\rhd$(\ref{K-1-gm})}

 \textsc{Normalize} the   weight $\widetilde w_\omega^{(I_1,\theta)}$; \hfill{$\rhd$(\ref{fuse-w})}

\For{$\ell\in\mathbb{L}_{1,g}$}{
 \textsc{Calculate} the MB parameter $r^{(\ell)}$;    \hfill{$\rhd$(\ref{r_l})}
  \textsc{Calculate} the MB parameter $p^{(\ell)}(\bx)$;  \hfill{$\rhd$(\ref{p_l})}
  }
  $\pi_\omega=\pi_\omega\cup\{(r^{(\ell)},p^{(\ell)}(\bx))\}_{\ell\in\mathbb{L}_{1,g}}$\;
  }
\textbf{return:} $\pi_\omega$.
\end{minipage}
\end{algorithm}
\subsection{Computational complexity}
The computational complexity of  P-GCI-MB is then analyzed  by comparison with the GCI fusion of the CPHD filter (GCI-CPHD) \cite{Battistelli,Uney-2}, considering both GM and SMC implementations. For the GM implementation,
 assume that  the probability   density of each Bernoulli component of each local MB filter is represented by a mixture of $J$ Gaussian components, and the location density of each local CPHD filter by a mixture of $n_{\max}J$ Gaussian components, where $n_{\max}$ denotes the maximum number of objects. Similarly, for the SMC implementation, assume that the probability   density of each Bernoulli component of each local MB filter is approximated by $N_{p}$ particles, while  the location density of each local CPHD filter is approximated by $n_{\max}N_{p}$ particles. Further, assume that the dimension of the single-object state $x$ is $N_d$.

 The computational complexity of P-GCI-MB mainly consists of the following four parts:
\begin{enumerate}
\item Computation of pairwise distances between hypothetic objects of two sensors:

$\bullet$ GM implementation -- $\mathcal{O}\left( L^2_{\max}J^2 N_d^3\right)$;

$\bullet$ SMC implementation -- $\mathcal{O}\left( L^2_{\max} N_{p}N_d^3\right)$;

where $L_{\max} =\max\{|\mathbb{L}_{1}|,|\mathbb{L}_{2}|\}$;

\item Finding the LIC: $\mathcal{O}\left(|E|+L_{\max}\alpha(|E|,L_{\max})\right)$, where $|E|=\frac{L_{\max}(L_{\max}-1)}{2}$ in the worst case;
\item Calculation of the fused single-object densities $p^{(\ell,\ell')}(\cdot)$:

$\bullet$ GM implementation --  $\mathcal{O}(N_{\mathcal{C}_{\text{I}}}[L^C_{\max}]^2J^2N_d^3)$;

$\bullet$ SMC implementation --  $\mathcal{O}(N_{\mathcal{C}_{\text{I}}}[L^C_{\max}]^2 N_{p} N_d^3)$;

where $L^C_{\max}=\max_{g:\mathcal{C}_g\in C_\text{I}}|\mathbb{L}_{1,g}|$ denotes the number of hypothetic objects in the largest cluster;

\item Calculation of the weight $\widetilde w_\omega^{(I_1,\theta)}$:  $\mathcal{O}\left(N_{\mathcal{C}_{\text{I}}} 2^{L^C_{\max}}L^C_{\max}!\right)$.

\end{enumerate}
It can be seen that the dominant part of the overall computational complexity is the last part which involves an exponential term and
a factorial calculation which grows exponentially with the number of hypothetic objects in the largest cluster $L^C_{\max}$.

In contrast, the computational complexity of the GM implementation of the GCI-CPHD fusion is about $\mathcal{O}\left(n_{\max}^2 J^2N_d^3\right)$ \cite{Battistelli}, while the SMC implementation is $\mathcal{O}\left(n_{\max}^2 N_{p}N_d^3\right)$. Comparing it with the complexity of P-GCI-MB, the computational costs of GCI-CPHD and P-GCI-MB become comparative when the number of hypothetic objects in the largest cluster $L^C_{\max}$ is small.

\section{Performance Assessment}
In this section, the performance of the proposed P-GCI-MB fusion is examined in two tracking scenarios in terms of the {Optimal Sub-Pattern Assignment} (OSPA) error \cite{MeMBer_Vo1}. The P-GCI-MB fusion is implemented using the GM approach proposed in Section IV. Since this paper does not focus on the problem of weight selection, we choose the Metropolis weights \cite{Xiao} in P-GCI-MB fusion for convenience. All performance metrics are averaged over 200 Monte Carlo (MC) runs.

The standard object and observation models \cite{LMB_Vo2} are used. The number of objects is time varying due to births and deaths. The single-object state $x^{k}=[p_x^k, p_y^k, v_x^k,  v_y^k]^\top$ at time $k$ is a vector composed of 2-dimensional position and velocity, and evolves to the next time according a linear Gaussian model,
\begin{equation}
f^k(\bx^k|\bx^{k-1}) = \mathcal{N}(\bx^k; \bF^k \bx^{k-1},\bQ^k)
\end{equation}
with its parameters given for a nearly constant velocity motion model:
\begin{equation}
 \bF^k =\left[
\begin{array}{cc}
I_2& \Delta I_2 \\ 0_2 & I_2
\end{array}
\right],\,\,{\bQ^k=\sigma_v^{2}\left[
\begin{array}{cc}
\frac{1}{4} \Delta^4 I_2& \frac{1}{2}\Delta^2 I_2 \vspace{1mm} \\  \frac{1}{3}\Delta^2 0_2 & \Delta^2 I_2
\end{array}
\right]}
\end{equation}
where $I_n$ and $0_n$ denote the $n\times n$ identity and zero matrices, $\Delta =1$\,s is the sampling period, and $\sigma_\nu^2=5\,\text{m/s}^{2}$ is the standard deviation of the process noise. The probability of object survival is $P_{S}^k=0.98$. Each sensor detects an object independently with probability $P_{D}^k=0.95$. The single-object observation model is linear Gaussian
\begin{equation} g^k(z^k|\bx^{k}) = \mathcal{N}(z^k; \bH^k \bx^{k},\bR^k)\end{equation}
with parameters $\bH^k=\left[
\begin{array}{cc}
I_2  & 0_2
\end{array}
\right]$ and $R^k=\sigma_\varepsilon^{2}I_2$,
where $\sigma_\varepsilon=10$\,m is the standard deviation of the measurement noise. The number of clutter reports in each scan is Poisson distributed with $\lambda=10$. Each clutter report is sampled uniformly over the whole surveillance region.

\subsection{Scenario 1}
We first study the estimation accuracy of the proposed P-GCI-MB algorithm by comparing it with the standard GCI-MB fusion in \cite{GCI-MB}. Considering the computational intractability of the standard GCI-MB fusion, here we use a simple scenario only involving two sensors and three objects on a two dimensional surveillance region $[-500\,\,500]$m $\times$ $[-500\,\,500]$m. To bound the complexity of the standard GCI-MB fusion, we also assume that the true times of births are known as prior knowledge, and hence, only at the first time step, a birth distribution with $3$ Bernoulli components is adopted. Then during the whole duration of $T=65$\,s, only three Bernoulli components are in running. The sketch map of scenario 1 is shown in Fig.~\ref{fig:performance_1} (a). Other parameters for GM implementation are set as follows. For both algorithms, the maximum number of Gaussian components for each Bernoulli components is $5$. The pruning and merging thresholds for Gaussian components are $\gamma_p=10^{-5}$ and $\gamma_m=4$, respectively. For P-GCI-MB, the GCI divergence threshold  $\gamma$ is set to 4.

The curves of OSPA errors for local filter, the standard GCI-MB fusion and P-GCI-MB fusion are plotted against time steps in Fig.~\ref{fig:performance_1} (b). The curve of the standard GCI-MB fusion can be viewed as the performance upper bound since it is a complete implementation without discarding any fusion hypothesis. We can see clearly that the curves of the standard GCI-MB fusion and P-GCI-MB fusion are almost identical, hence verifying the accuracy of the adopted independence assumption. Besides, both of the two fusion methods outperform the local filter significantly due to the information exchange and posterior fusion of two sensors.
\begin{figure}[!t]
\begin{minipage}[htb]{0.495\linewidth}
\centering
\includegraphics[width=4.5cm]{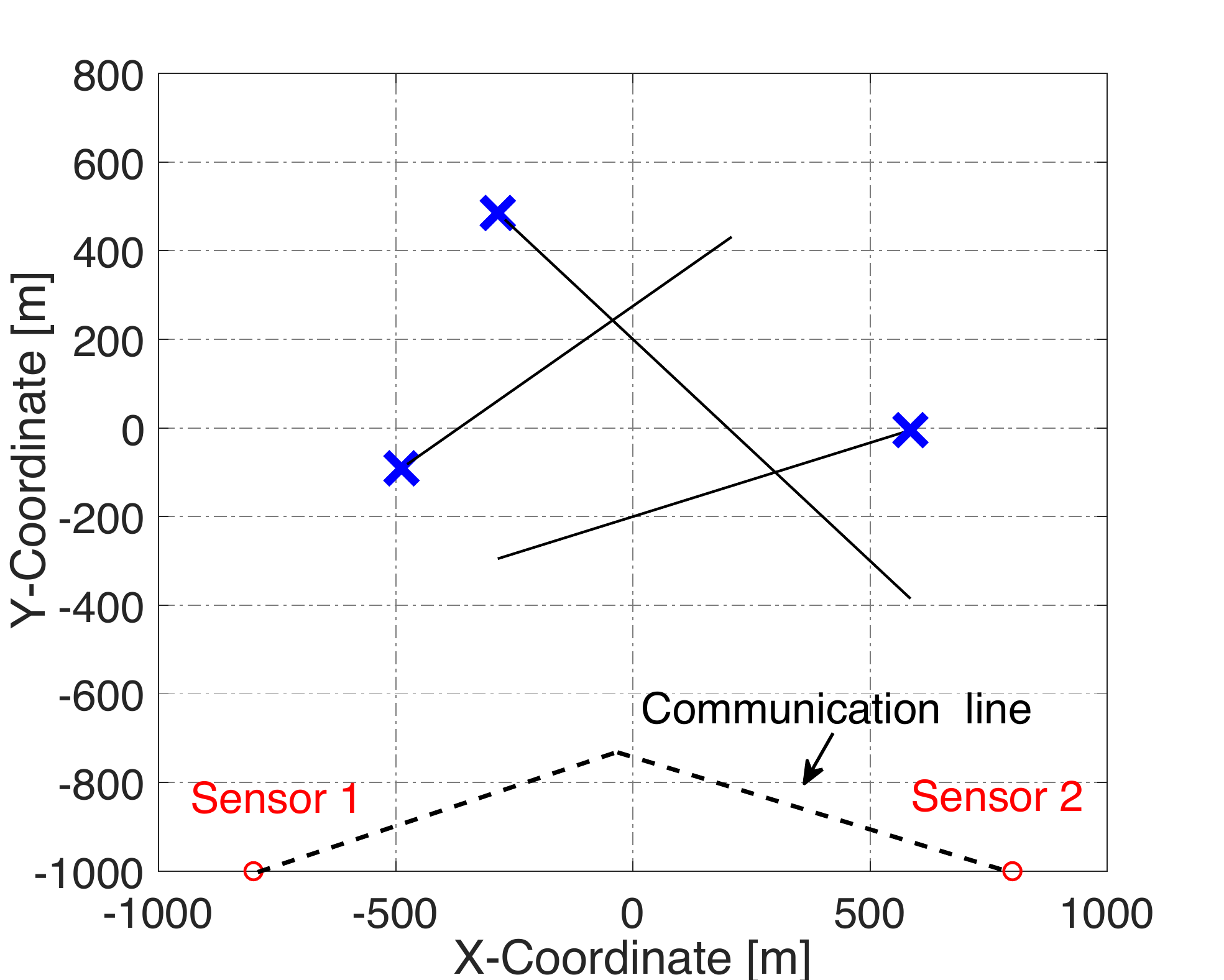}
 \centerline{\small{(a)} }\medskip
\end{minipage}
\hfill
\begin{minipage}[htb]{0.495\linewidth}
  \centering
  \centerline{\includegraphics[width=4.5cm]{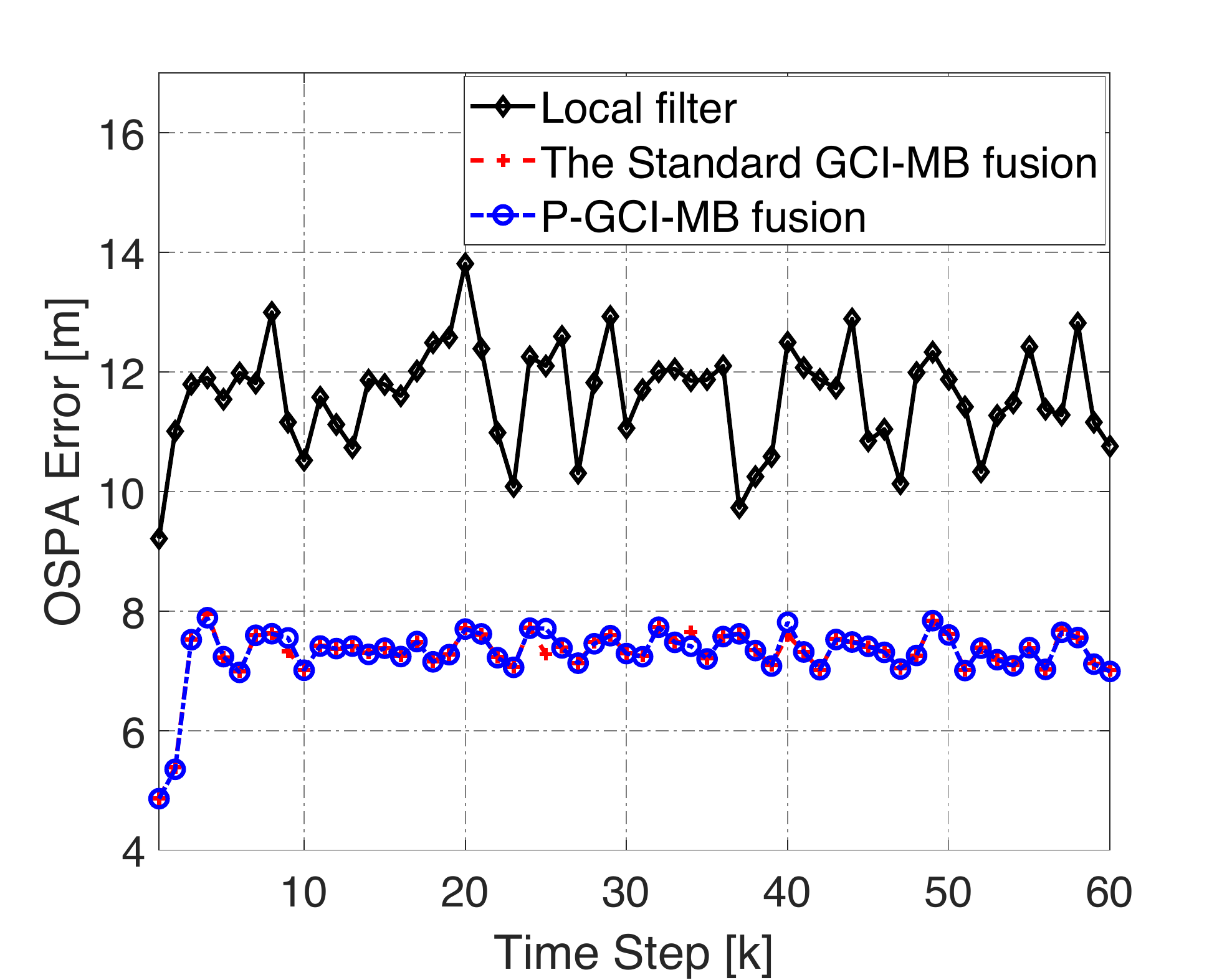}}
  \centerline{\small{(b)} }\medskip
\end{minipage}
\caption{(a) The scenario of the distributed sensor network with two sensors tracking three objects. The initial positions of the objects are indicated by crosses. (b) The curves of OSPA errors of local filter, the standard GCI-MB fusion and the P-GCI-MB fusion.\label{fig:performance_1}}
\end{figure}

\subsection{Scenario 2}
To examine the performance of the P-GCI-MB fusion in challenging scenarios, we consider a distributed sensor network with six sensors tracking forty objects in a $[-2000\,\,2000]$\,m $\times$ $[-2000\,\,2000]$\,m  surveillance region. The sketch map of this scenario is given in  Fig.~\ref{fig:scenario_2} (a), and the duration of tracking is $T=200$\,s. Further, Fig. \ref{fig:scenario_2} (b) shows superposition of all observations acquired by Sensor 5 during the period of 1--200 s.
\begin{figure}[!t]
\begin{minipage}[htb]{0.49\linewidth}
  \centering
  \centerline{\includegraphics[width=4.5cm]{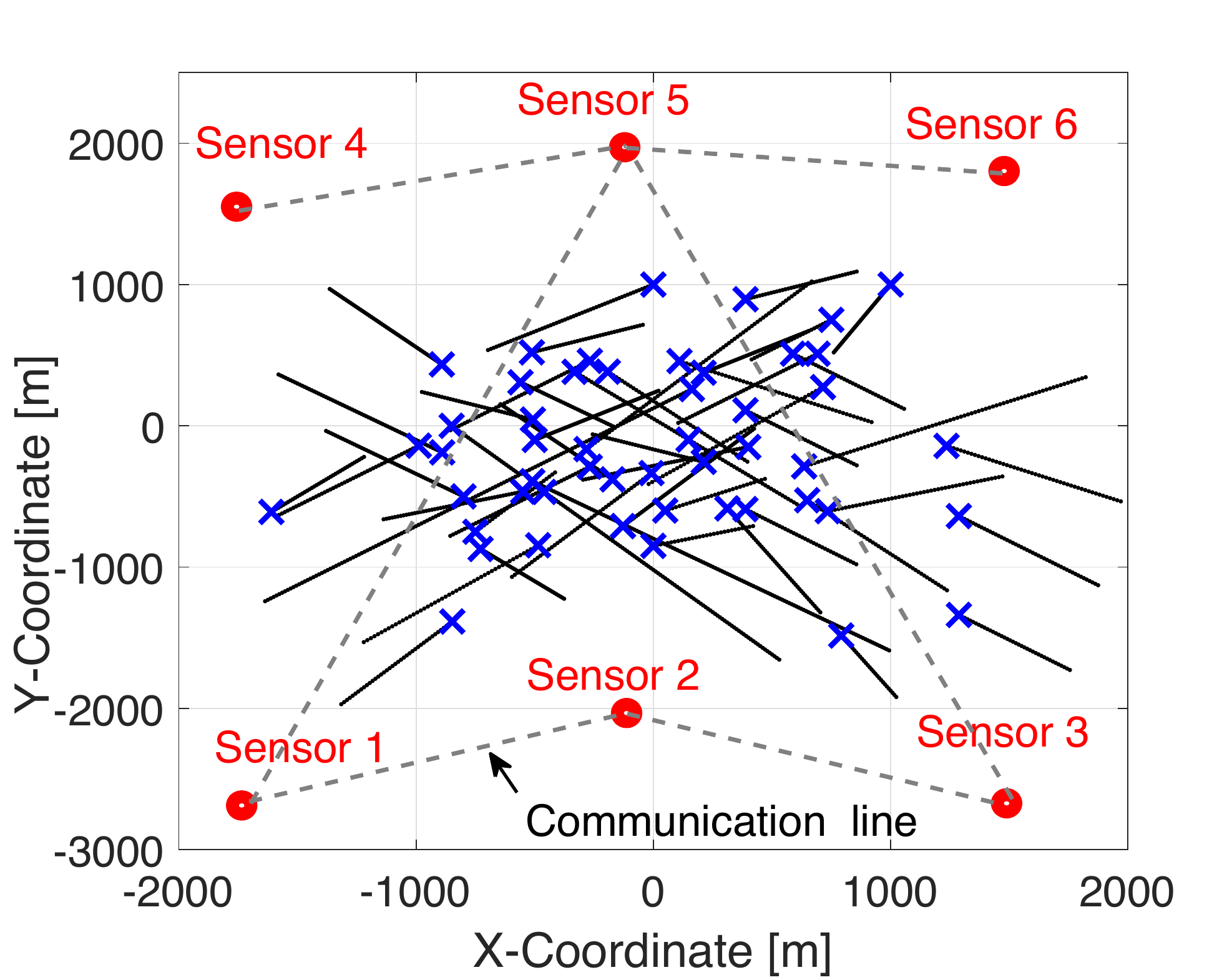}}
  \centerline{\small{(a)} }\medskip
\end{minipage}
\hfill
\begin{minipage}[htb]{0.49\linewidth}
  \centering
  \centerline{\includegraphics[width=4.5cm]{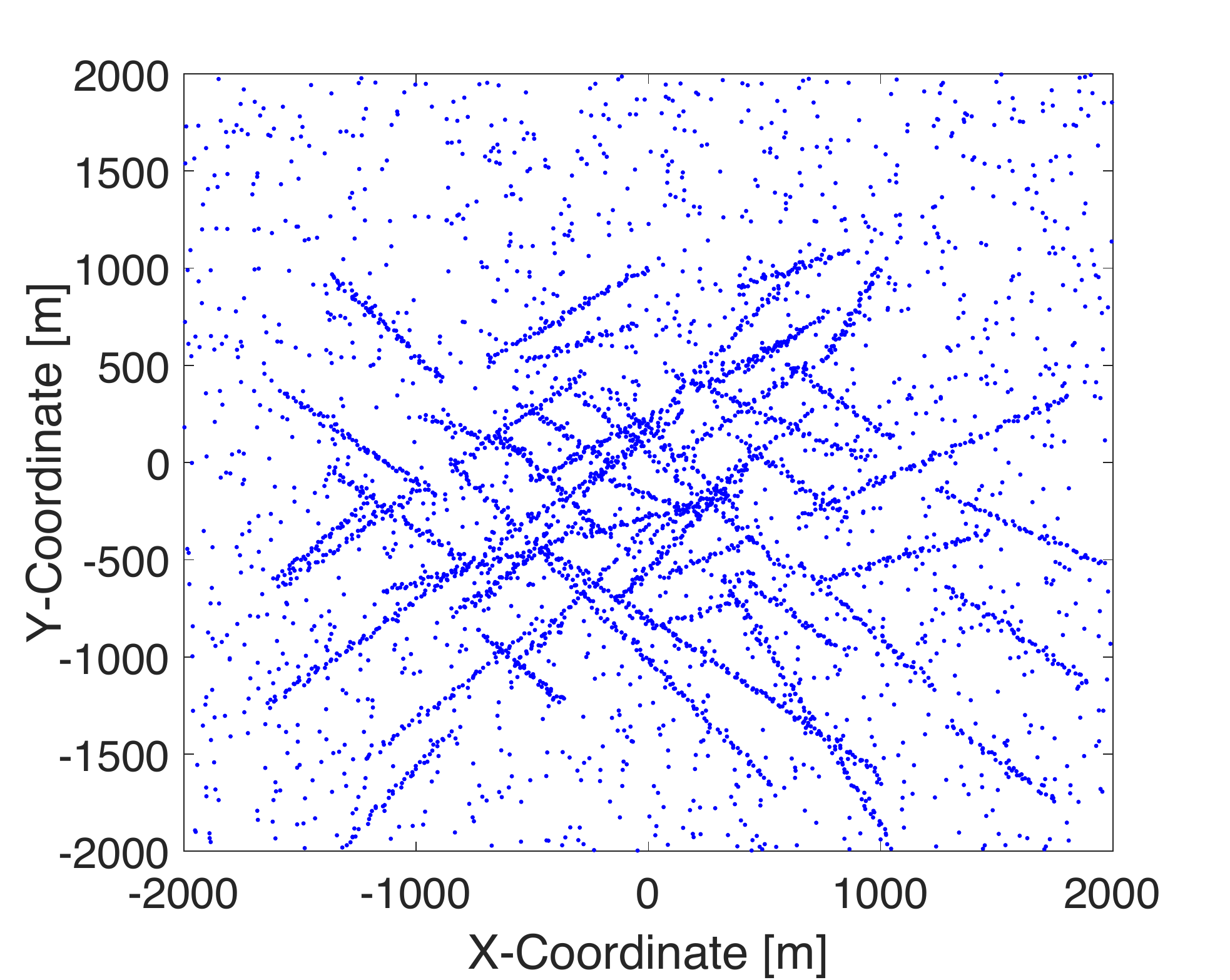}}
  \centerline{\small{(b)}}\medskip
\end{minipage}
\caption{(a) The scenario of a distributed sensor network with three sensors tracking 50 objects, where the initial positions of objects are indicated by crosses. Each sensor can only exchange posteriors with its neighbor(s); (b) Superposition of all observations acquired by Sensor 5 during  1--200 s, where each blue dot denotes either the observation of an object or a clutter report.\label{fig:scenario_2} }
\end{figure}

In order to demonstrate both the computational efficiency and estimation accuracy of the proposed P-GCI-MB method, here we compare it with the GCI fusion with PHD filter (GCI-PHD) in \cite{Uney-2} and the GCI fusion with CPHD filter (GCI-CPHD) in \cite{Battistelli}.
The PHD filter \cite{PHD-Vo}, CPHD filter \cite{Vo-CPHD} and cardinality balanced MB filter \cite{MeMber_Vo2} are chosen as the local filter for GCI-PHD fusion, GCI-CPHD fusion and P-GCI-MB fusion, respectively. Since the objects appear at unknown positions and unknown time, the MB filters adopt the adaptive birth procedure introduced in \cite{LMB-Reuter}, and the PHD/CPHD filters use the adaptive birth distribution introduced in \cite{Ristic_PHD}.

The implementation parameters of different algorithms are chosen as follows. For the P-GCI-MB fusion,  the maximum number of Bernoulli components  is set to $100$; the truncation threshold for Bernoulli components is $\gamma_t=10^{-4}$;
the maximum number, the pruning and merging thresholds of Gaussian components, and the GCI divergence threshold are set to be the same as Scenario 1.
For the GCI-CPHD and GCI-PHD fusion, the maximum number of Gaussian components is set to $150$; the pruning and merging thresholds for Gaussian components are the same as the P-GCI-MB fusion.

Regarding the sensor network topological structure, each sensor is linked with its neighbor(s) via the communication line as shown in Fig. \ref{fig:scenario_2} (a). Through the communication line, sensors can only exchange posteriors with their neighbor(s). Therefore, for example, sensor 5 performs fusion with five posteriors from sensors 1, 3, 4 and 6, and the local filter by sequentially applying the pairwise fusion four times.
\begin{figure}[!t]
\begin{minipage}[htb]{0.49\linewidth}
  \centering
  \centerline{\includegraphics[width=4.5cm]{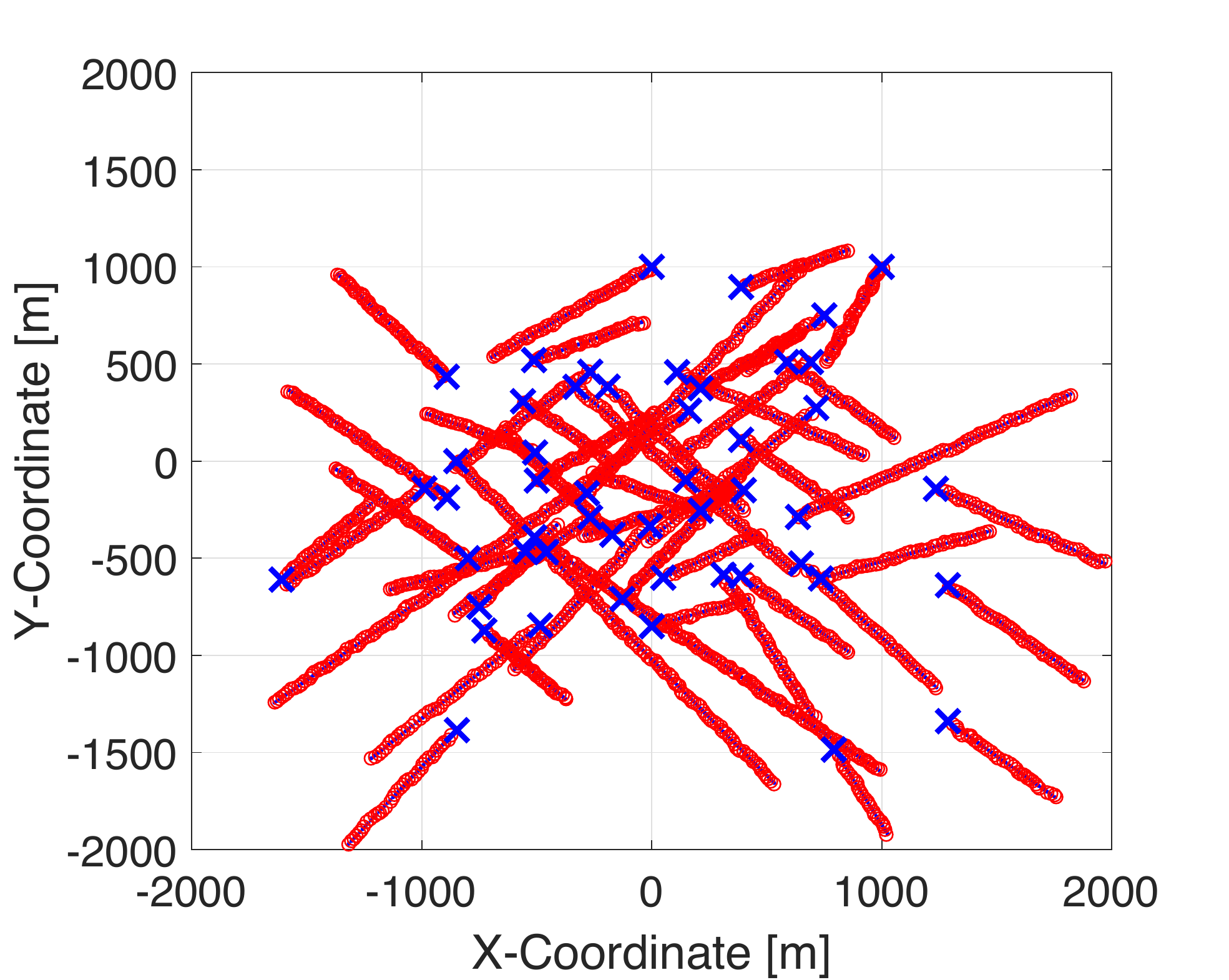}}
  \centerline{\small{(a)} }\medskip
\end{minipage}
\hfill
\begin{minipage}[htb]{0.49\linewidth}
  \centering
  \centerline{\includegraphics[width=4.5cm]{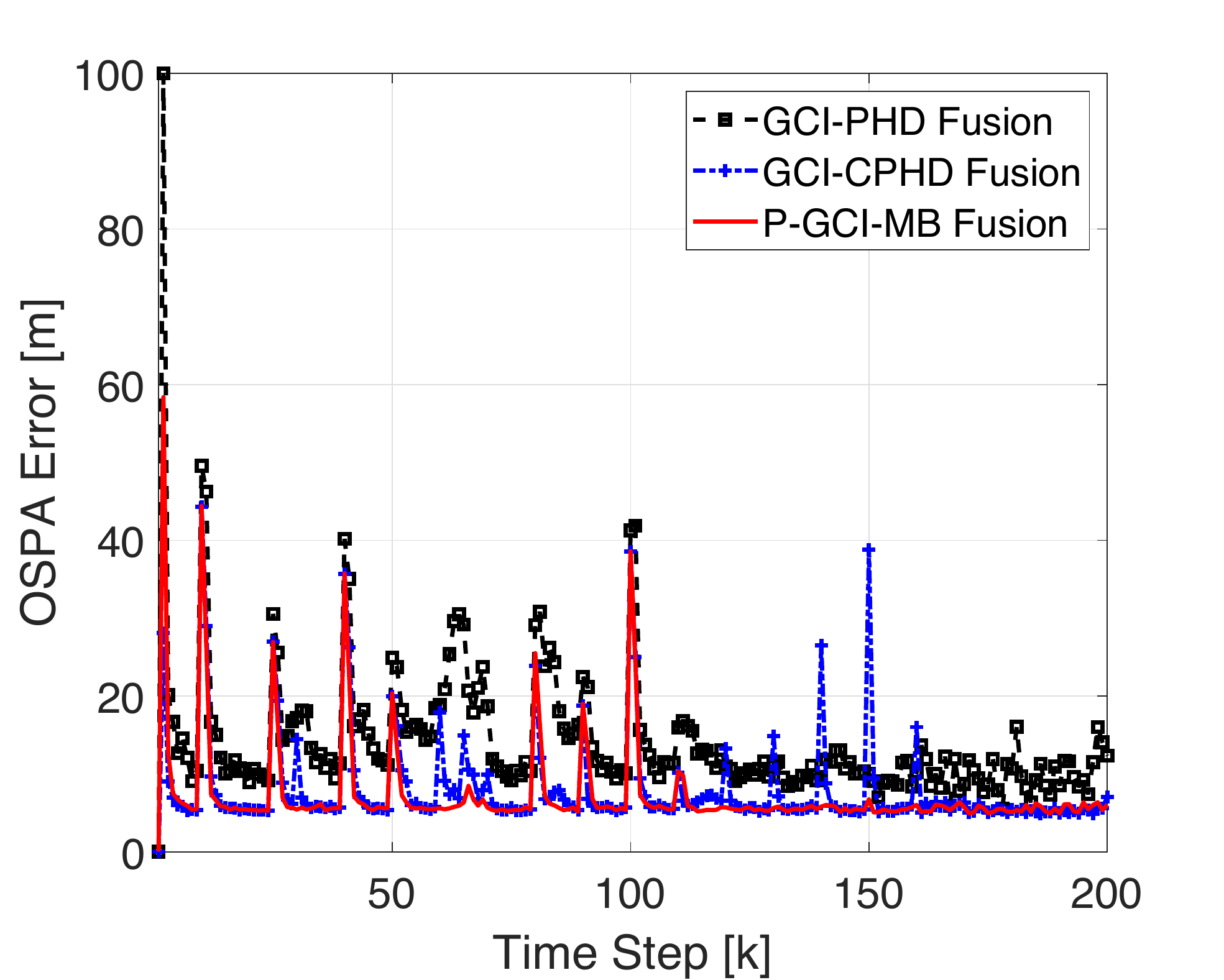}}
  \centerline{\small{(b)}}\medskip
\end{minipage}
\caption{(a) Output of P-GCI-MB fusion for a single MC run, with object estimates indicated by red circles; (b) Tracking performance of P-GCI-MB, GCI-PHD and GCI-CPHD fusion algorithms at Sensor 5.\label{fig:ospa_scenario_2} }
\end{figure}
\begin{figure}[!t]
\begin{minipage}[htb]{0.49\linewidth}
  \centering
  \centerline{\includegraphics[width=4.5cm]{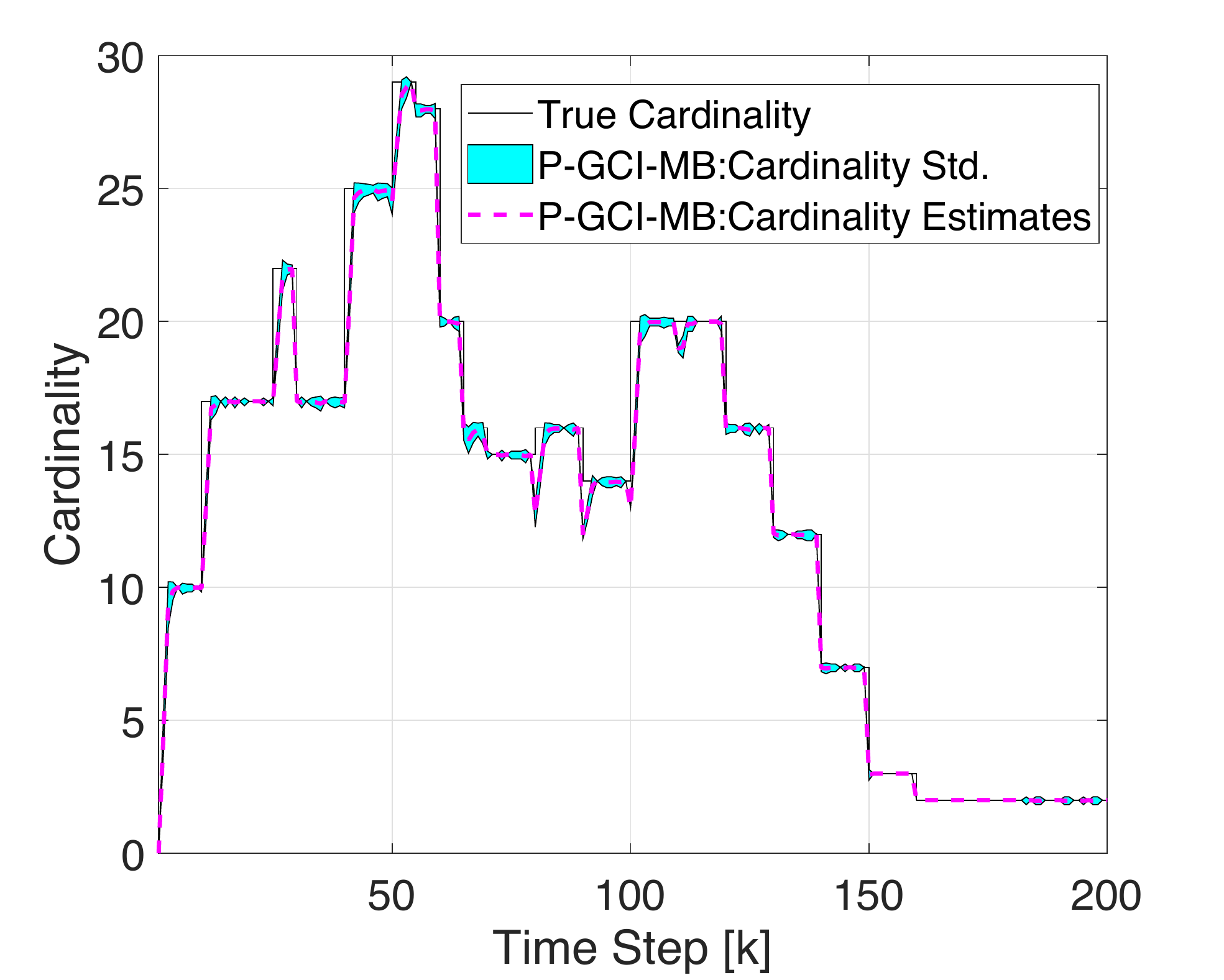}}
  \centerline{\small{(a)} }\medskip
\end{minipage}
\hfill
\begin{minipage}[htb]{0.49\linewidth}
  \centering
  \centerline{\includegraphics[width=4.5cm]{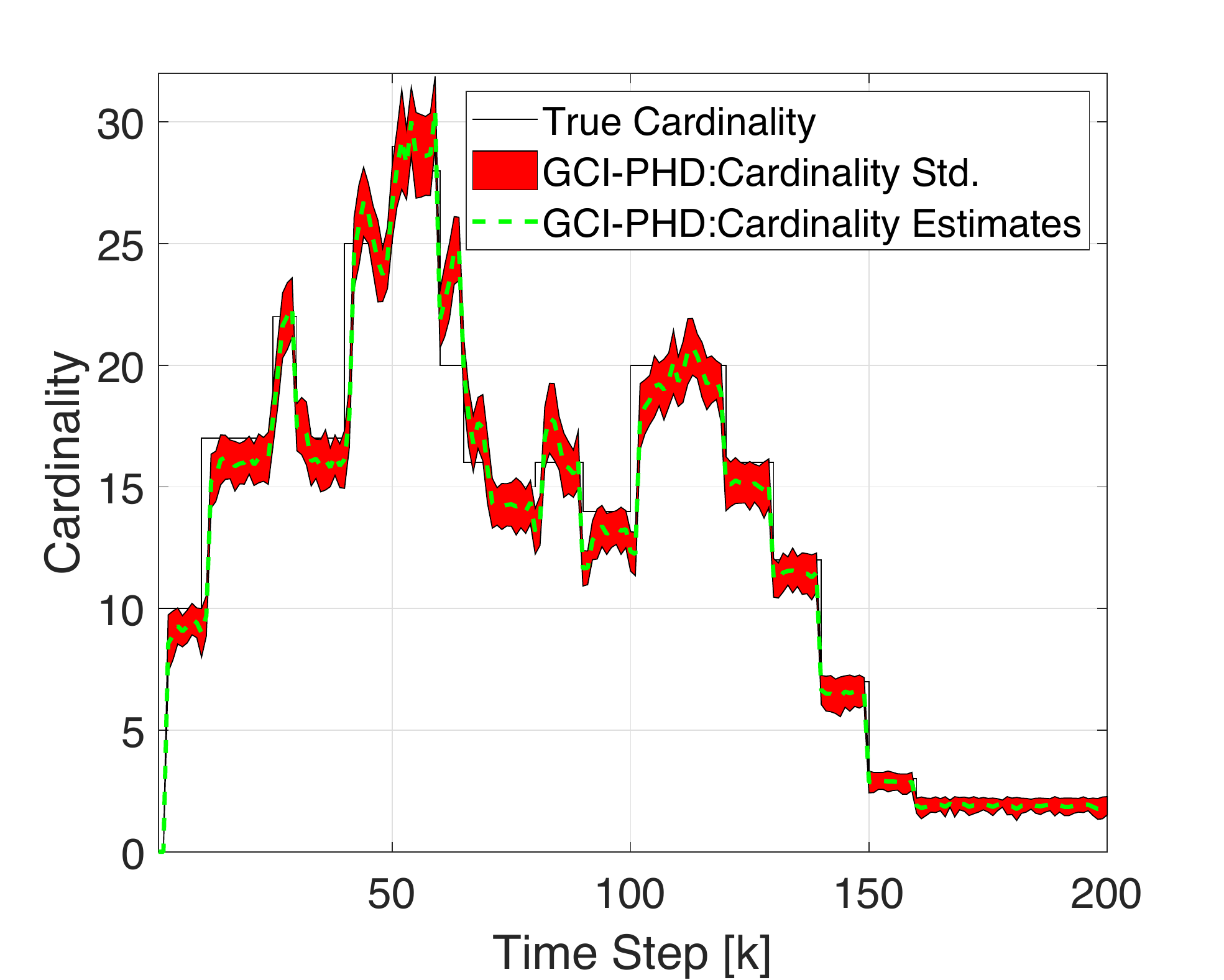}}
  \centerline{\small{(b)}}\medskip
\end{minipage}
\hfill
\begin{minipage}[htb]{0.49\linewidth}
  \centering
  \centerline{\includegraphics[width=4.5cm]{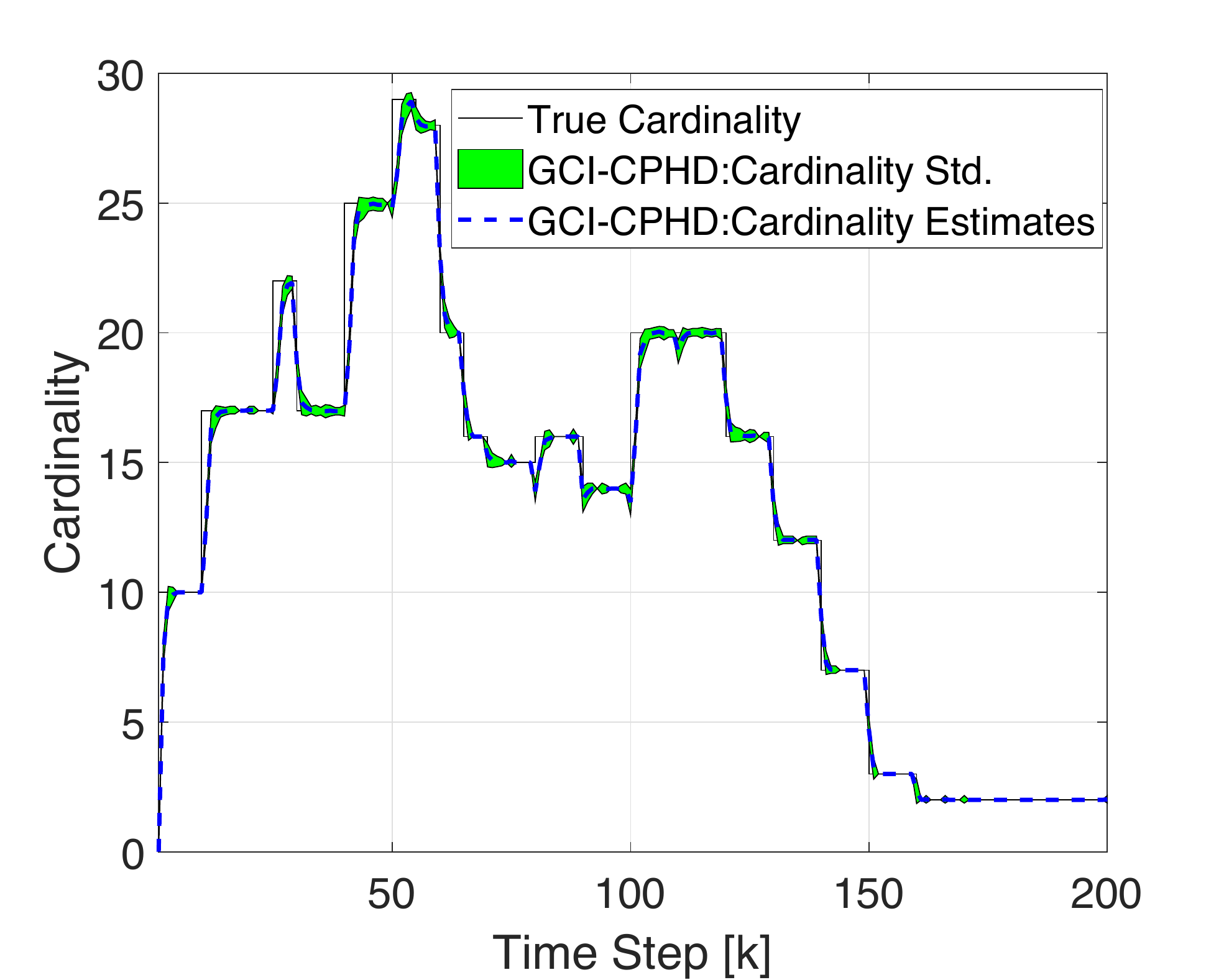}}
  \centerline{\small{(c)}}\medskip
\end{minipage}
\hfill
\begin{minipage}[htb]{0.49\linewidth}
  \centering
  \centerline{\includegraphics[width=4.5cm]{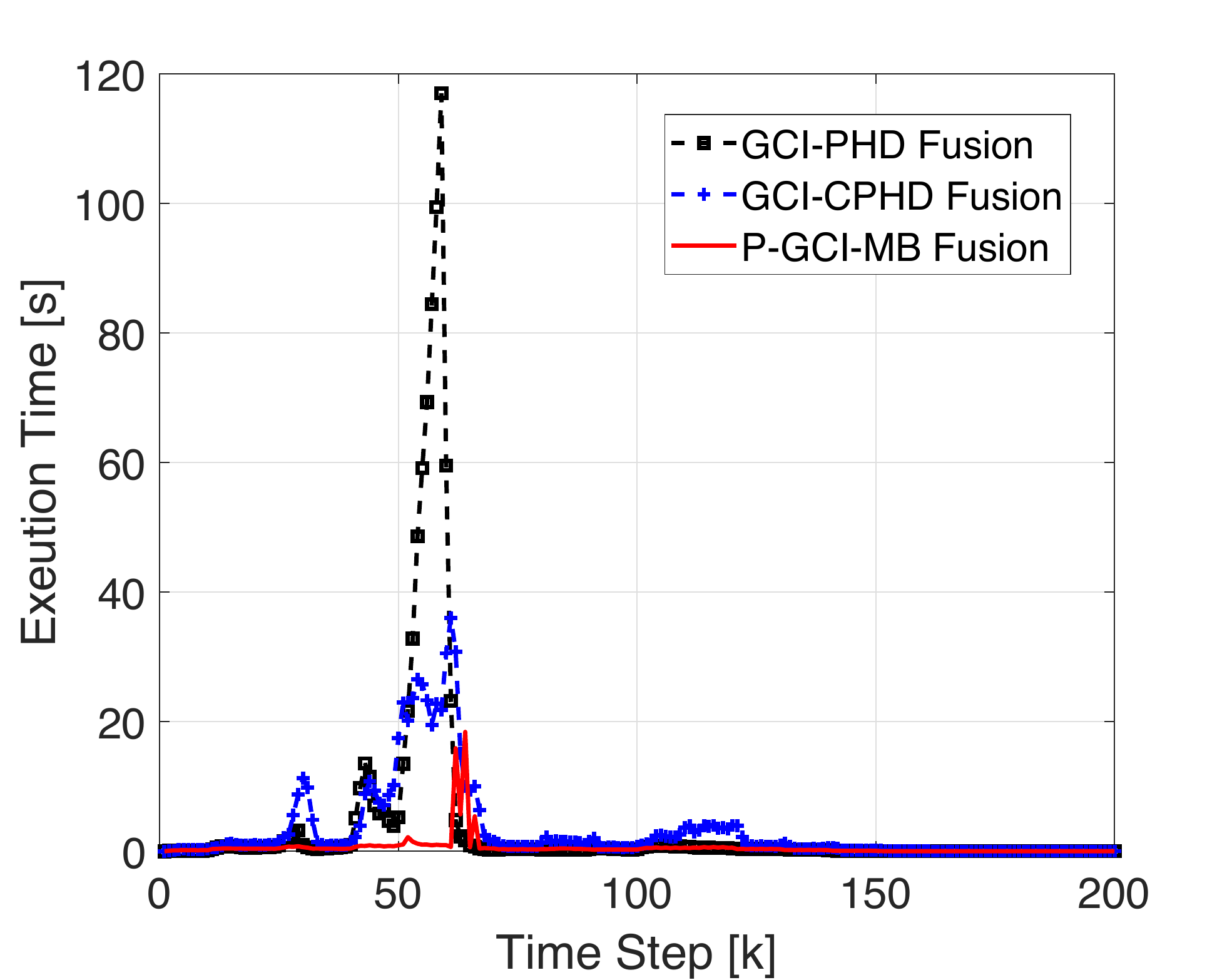}}
  \centerline{\small{(d)}}\medskip
\end{minipage}
\caption{Cardinality statistics at Sensor 5: (a) P-GCI-MB fusion; (b) GCI-PHD fusion; (c) GCI-CPHD fusion; (d) execution times.\label{fig:cardinality_time_scenario_2} }
\end{figure}
Fig. \ref{fig:ospa_scenario_2} (b) first shows the outputs of P-GCI-MB for a single MC run. It can be seen that P-GCI-MB performs accurately and consistently in the sense that it maintains locking on all objects and correctly estimates object positions for the entire scenario. Next, the OSPA errors for tracking results returned by the fusion algorithms performed at Sensor 5 are shown in Fig.~\ref{fig:ospa_scenario_2} (b). The OSPA error curves clearly demonstrate the performance difference among the P-GCI-MB, GCI-CPHD and GCI-PHD fusions. Generally speaking, both P-GCI-MB and GCI-CPHD outperform GCI-PHD, with much lower OSPA errors. Moreover, at the times of object death including 30\,s, 60\,s, 65\,s, 70\,s, 120\,s, etc., the OSPA errors of the P-GCI-MB fusion converge very fast, while the fusion performance of GCI-CPHD decreases a lot and hence a ``peak'' arises at its curve. After each transients, the performance of P-GCI-MB and GCI-CPHD are at the same level.

Figs.~\ref{fig:cardinality_time_scenario_2} (a)\,-\,(c) present the cardinality estimates and the corresponding standard deviations (Std.) returned by P-GCI-MB, GCI-CPHD and GCI-PHD at Sensor 5. It shows that the cardinality estimates given by both P-GCI-MB and GCI-CPHD are more accurate with less variations (higher level of confidence) than the GCI-PHD fusion. Further, when objects disappear, the convergence rate of cardinality estimation curve for the P-GCI-MB is faster than the GCI-CPHD.

Additionally, the comparison of the averaged execution time per time step among the P-GCI-MB, GCI-PHD and GCI-CPHD fusion methods are shown in Fig.~\ref{fig:cardinality_time_scenario_2} (d). It can be seen that the average execution time of P-GCI-MB fusion is less than GCI-CPHD and GCI-PHD over almost all time steps. Furthermore, the time gaps among P-GCI-MB, GCI-CPHD and GCI-PHD fusions become more distinct when there are more than ten objects appear in the scenario.
\section{Conclusion}
This paper investigates an efficient distributed  fusion with multi-Bernoulli filters based on generalized Covariance Intersection. By discarding the inter-cluster hypotheses  in the original fused posterior, the fused posterior is simplified to multiple independent fused ones with a  smaller number of hypothetic objects, which  leads to a significant computational reduction as well as an appealing parallel calculation.  The characterization of the $L_1$-error for the  approximation is also analysed.  The efficiency of the proposed algorithm  is verified via numerical results.
\appendices
\section{Proof of Proposition 1}
\begin{proof}
Recall that $\widetilde\pi'_\omega(X)$, which is the un-normalized GMB fused density after discarding all inter-cluster hypotheses,  is computed by (\ref{unnormalized_trun}).
For each hypothesis involved in $\widetilde\pi'_\omega(X)$, i.e., $(I_1,\theta)\in\mathcal{F}(\mathbb{L}_1)\times\Theta-\mathbb{D}$, we have $\theta(\ell)=\theta_g(\ell)$ for each $\ell\in\mathbb{L}_{1,g}$, $g\in\{g:\mathcal{C}_g\in C_\text{I}\}$. Then, utilzing Definition 1,  the space $\mathcal{F}(\mathbb{L}_1)\times\Theta-\mathbb{D}$ can be further represented as
\begin{equation}\label{union-space}
\mathcal{F}(\mathbb{L}_1)\times\Theta-\mathbb{D}=\bigcup_{g:\mathcal{C}_g\in C_\text{I}} \mathcal{F}(\mathbb{L}_{1,g})\times\Theta_g.
\end{equation}
According to (\ref{Phi2_1}), the parameter $Q_1^{I_1}$ in (\ref{fuse-w_p}) can be further rewritten as
\begin{equation}\label{factorization-Q1}
\begin{split}
Q_{1}^{I_1}=&\prod_{g:\mathcal{C}_g\in C_\text{I}} \prod_{\ell\in\mathbb{L}_{1,g}\backslash (I_1\cap \mathbb{L}_{1,g})}(1-r_s^{(\ell)})\prod_{\ell'\in I_1\cap\mathbb{L}_{1,g}} r_s^{(\ell)}\\ =&\prod_{g:\mathcal{C}_g\in C_\text{I}}Q_{1,g}^{I_1\cap\mathbb{L}_{1,g}},
\end{split}
\end{equation}
where $Q_{1,g}^{I_1\cap\mathbb{L}_{1,g}}$ is the corresponding parameter of the $g$th cluster at sensor 1;
while, similarly, $Q_2^{\theta(I_1)}$ in (\ref{fuse-w_p}) can be rewritten as
\begin{align}\label{factorization-Q2-1}
Q_{2}^{\theta(I_1)}=\prod_{g:\mathcal{C}_g\in C_\text{I}}Q_{2,g}^{\theta(I_1)\cap\mathbb{L}_{2,g}}=\prod_{g:\mathcal{C}_g\in C_\text{I}}Q_{2,g}^{\theta(I_1\cap\mathbb{L}_{1,g})}.
\end{align}
Since  we have $\theta(\ell)=\theta_g(\ell)$ for each $\ell\in I_1\cap\mathbb{L}_{1,g}$, hence,
\begin{equation}\label{factorization-Q2}
Q_{2}^{\theta(I_1)}=\prod_{g:\mathcal{C}_g\in C_\text{I}}Q_{2,g}^{\theta_g(I_1\cap\mathbb{L}_{1,g})}.
\end{equation}
Similarly, the following holds:
\begin{equation}\label{p_g}
p^{(\ell,\theta)}(x)=p^{(\ell,\theta_g)}(x), \ell\in I_1\cap\mathbb{L}_{1,g}, g\in\{g:\mathcal{C}_g\in C_\text{I}\},
\end{equation}
and hence the term,  ${\prod}_{\ell\in I_1}Z_\omega^{(\ell,\theta)}$ in (\ref{fuse-w_p}) can be factorized as
\begin{equation}\label{factorization-Z}
\prod_{\ell\in I_1}Z_\omega^{(\ell,\theta)}=\prod_{g:\mathcal{C}_g\in C_\text{I}}\prod_{\ell\in I_1\cap\mathbb{L}_{1,g}}Z_{\omega,g}^{(\ell,\theta_g)},
\end{equation}
where $Z_{\omega,g}^{(\ell,\theta_g)}$ is given in $(\ref{Z_w-g})$.

As a result, based on (\ref{factorization-Q1}), (\ref{factorization-Q2}) and (\ref{factorization-Z}), the un-normalized weight in (\ref{fuse-w_p}) yields
\begin{equation}\label{factorization-w}
\widetilde{w}_\omega^{(I_1,\theta)}=\prod_{g:\mathcal{C}_g\in C_\text{I}} \widetilde{w}_\omega^{(I_1\cap\mathbb{L}_{1,g},\theta_g)},
\end{equation}
where $\widetilde{w}_\omega^{(I_1\cap\mathbb{L}_{1,g},\theta_g)}$ is given in (\ref{fuse-widetilde-w_p-g}). Then, using (\ref{union-space}), (\ref{p_g}) and (\ref{factorization-w}),  (\ref{unnormalized_trun}) can be further simplified as (\ref{cluster-numerator}), where $\widetilde\pi_{\omega,g}(\cdot)$ is the un-normalized fused density of the $g$th cluster,
\begin{equation}
\begin{split}
 \!\!\!\widetilde\pi_{\omega,g}(\{\bx_1,\!\cdots\!,\bx_n\})\!=\!&\sum_{\sigma}
\sum_{(I_{1,g},\theta)\in \mathcal{F}_{n}\!(\mathbb{L}_{1,g}\!) \!\times\! \Theta_{g}({I_{1,g}})}\!\!\!\widetilde w_\omega^{(I_{1,g},\theta_g)}\\
&\cdot{\prod}_{i=1}^{n}p_\omega^{([I_{1,g}]^v\!(i), \theta_g)}\!(x_{\sigma(i)}).
\end{split}
\end{equation}

\begin{figure*}
\begin{equation}\label{cluster-numerator}
\begin{split}
&\widetilde{\pi}'_\omega(\!\left\{\bx_1,\!\cdots\!,\bx_n\right\})
=\sum_\sigma\sum_{(I_1,\theta)\in \mathcal{F}_n(\mathbb{L}_1) \times \Theta_{I_1} -\mathbb{D}}\widetilde{w}_\omega^{(I_1,\theta)}{\prod}_{i=1}^{n}p_\omega^{(\bII_1^v(i),\theta)}(\bx_{\sigma(i)}\!)\\
=&\sum_{\biguplus_{g:\mathcal{C}_g\in C_\text{I}}\{x_{g_1},\cdots,x_{g_{n_g}}\}=\{x_1,\cdots,x_n\}}\sum_{\substack{\sigma_g:\{1,\cdots,n_g\}\rightarrow\{1,\cdots,n_g\},\\g:\mathcal{C}_g\in C_\text{I}}}\sum_{\substack{(I_{1,g},\theta_g)\in\mathcal{F}_{n_g}(\mathbb{L}_{1,g}\times\Theta_g),\\g:\mathcal{C}_g\in C_\text{I}}} \prod_{g:\mathcal{C}_g\in C_\text{I}} \widetilde{w}_\omega^{(I_{1,g},\theta_g)}\prod_{j=1}^{n_g}p_\omega^{([I_{1,g}]^v(j),\theta_g)}(x_{\sigma(g_j)})\\
=&\sum_{\biguplus_{g:\mathcal{C}_g\in C_\text{I}}\{x_{g_1},\cdots,x_{g_{n_g}}\}=\{x_1,\cdots,x_n\}}\prod_{g:\mathcal{C}_g\in C_\text{I}}\sum_{\sigma_g:\{1,\cdots,n_g\}\rightarrow\{1,\cdots,n_g\}}\sum_{\substack{(I_{1,g},\theta_g)}\in\mathcal{F}_{n_g}(\mathbb{L}_{1,g})\times\Theta_g} \widetilde{w}_{\omega,g}^{(I_{1,g},\theta_g)}\prod_{j=1}^{n_g}p_{\omega,g}^{([I_{1,g}]^v(j),\theta_g)}(x_{g_{\sigma(j)}})\\
=&\sum_{\biguplus_{g:\mathcal{C}_g\in C_\text{I}}X_g=\{x_1,\cdots,x_n\}}\prod_{g=1}^{N_{\mathcal{C}}}\widetilde\pi_{\omega,g}(X_g)
\end{split}
\end{equation}
\rule{\textwidth}{0.35mm}
\end{figure*}

Further, based on (\ref{cluster-numerator}), the normalized constant $\eta_\omega$  can be computed by:
\begin{equation}\label{cluster-C}
\begin{split}
\eta_\omega&=\sum_{I_1\in \mathcal{F}(\mathbb{L}_1)}\sum_{\theta\in\Theta_{I_1}}\widetilde{w}_\omega^{(I_1,\theta)}\,\,\,\,\,\\
&=\prod_{g:\mathcal{C}_g\in C_\text{I}}{\sum}_{I_{1,g}\in \mathcal{F}(\mathbb{L}_{1,g})}{\sum}_{\theta_g\in\Theta_g(I_1)}\widetilde{w}_\omega^{(I_1,\theta)}\\
&= \prod_{g:\mathcal{C}_g\in C_\text{I}}\eta_{\omega,g}.
\end{split}
\end{equation}
Hence, substituting (\ref{cluster-numerator}) and (\ref{cluster-C}) into (\ref{GMB-fused-density-discarded}), we have
\begin{equation}\label{G-GCI-MB-proof}
\begin{split}
&\pi'_\omega(X)\\
=&\sum_{\biguplus_{g:\mathcal{C}_g\in C_\text{I}}X_g=X} \prod_{g:\mathcal{C}_g\in C_\text{I}}\! \widetilde\pi_{\omega,g}(X_g)/\eta_{\omega,g}\\
=&\sum_{\biguplus_{g:\mathcal{C}_g\in C_\text{I}}X_g=X} \prod_{g:\mathcal{C}_g\in C_\text{I}}\pi_{\omega,g}(X_g).
\end{split}
\end{equation}
Hence, the Proposition holds.
\end{proof}
\section{Proof of Proposition 2}
\begin{proof}
Based on (\ref{union-space}), we can further obtain that
\begin{equation}
N_{H}'=|\mathcal{F}(\mathbb{L})\times\Theta-\mathbb{D}|=\prod_{g:\mathcal{C}_g\in C_\text{I}} |\mathcal{F}(\mathbb{L}_{1,g})\times\Theta_g|.
\end{equation}

Consider a function
\begin{equation}\label{function-standard}
\phi(a_1,\cdots,a_{N})=\prod_{i=1}^N a_i-\sum_{i=1}^N  a_i,
 \end{equation}
with ${N} \geq2$.

It can be easily checked that if $a_1=\cdots=a_{N} =2$, the following holds:
\begin{equation}
\phi(a_1,\cdots,a_N)\geq 0
\end{equation}
where  the equality holds up if and only if ${N} =2$.

The partial derivative of $\phi(a_1,\cdots,a_{N})$ with respective to $a_i$ is computed by
\begin{equation}
\frac{d\,\phi}{d\,a_i}=a_i(\prod_{i'\in\{1,\cdots,{N} \}\backslash \{i\}}a_{i'}-1), i=1,\cdots,N
\end{equation}
Hence, for  $a_i> 2, i=1,\cdots,N$,  we have
\begin{align}
\frac{d\,\phi}{d\,a_i}>0,
\end{align}
and thus the following holds:
\begin{equation}
\phi(a_1,\cdots,a_{N})> 0.
\end{equation}

Define each item $a_i$ in (\ref{function-standard}) as $\triangleq|\mathcal{F}(\mathbb{L}_{1,g})\times\Theta_g|$ for each $g:\mathcal{C}_g\in\mathbb{C}_\text{I}$.
It can be easily checked that  $|\mathcal{F}(\mathbb{L}_{1,g})\times\Theta_g|\geq 2$ if $\mathbb{L}_{1,g}\geq 1$, where the equality holds if and only if $\mathbb{L}_{1,g}=1$ and $\mathbb{L}_{2,g}=1$. Hence, we have
\begin{equation}
\prod_{g:\mathcal{C}_g\in\mathbb{C}_\text{I}} |\mathcal{F}(\mathbb{L}_{1,g})\times\Theta_g|-\sum_{g:\mathcal{C}_g\in\mathbb{C}_\text{I}} |\mathcal{F}(\mathbb{L}_{1,g})\times\Theta_g|\geq 0
\end{equation}
if $\mathbb{L}_{1,g}\geq 1, g:\mathcal{C}_g\in\mathbb{C}_\text{I}$, where the equality holds if and only if $\mathbb{L}_{1,g}=1$, $\mathbb{L}_{2,g}=1$ and $N_{\mathcal{C}_\text{I}}=2$.
\end{proof}

\section{Proof of Proposition 3}
\begin{proof}Based on the definition of $L_1$-norm, we have,
\begin{equation}
\notag\begin{split}
&\,\,\,\,\,\,\, \left\| \pi_{\omega}(\cdot)-\pi'_{\omega}(\cdot)\right\|_1\\
&\leqslant \int \left| \sum_{\sigma}\!\sum_{(I_1,\theta)\in \mathbb{H}-\mathbb{D}}\!\left(\!\frac{\widetilde{w}_\omega^{(I_1,\theta)}}{\eta_\omega}\!-\!\frac{\widetilde{w}_\omega^{(I_1,\theta)}}{\eta'_\omega}\!\right)\!\prod_{i=1}^{n}p_\omega^{([I_1]^v(i),\theta)}(\bx_{\sigma(i)})\right| \delta X\\
\end{split}
\end{equation}

\begin{equation}
\begin{split}
&\,\,\,\,\,\,\,\,\,+\int \left| \sum_{\sigma}\sum_{(I_1,\theta)\in \mathbb{D}}\frac{\widetilde{w}_\omega^{(I_1,\theta)}}{\eta'_\omega}\,\prod_{i=1}^{n}p_\omega^{([I_1]^v(i),\theta)}(\bx_{\sigma(i)})\right| \delta X\\
& =\sum_{(I_1,\theta)\in \mathbb{H}-\mathbb{D}}\left|\frac{\widetilde{w}_\omega^{(I_1,\theta)}}{\eta_\omega}-\frac{\widetilde{w}_\omega^{(I_1,\theta)}}{\eta'_\omega}\right|+\sum_{(I_1,\theta)\in \mathbb{D}}\frac{\widetilde{w}_\omega^{(I_1,\theta)}}{\eta'_\omega}\\
&=1-\frac{\eta'_\omega}{\eta_\omega}+\frac{\eta_\omega-\eta'_\omega}{\eta'_\omega}\\
&=2\frac{(\eta_\omega-\eta'_\omega)}{\eta_\omega}=2\sum_{(I,\theta)\in\mathbb{D}} \widetilde{w}_\omega^{(I_1,\theta)}\big/ \eta_\omega\\
&=2\sum_{(I_1,\theta)\in\mathbb{D}} \!\!\left(Q_1^{I_1}\right)^{\omega_1}\!\!\left(Q_2^{\theta(I_1)}\right)^{\omega_2}\!\!\exp(-\!\sum_{\ell\in I_1}d(\ell,\theta(\ell)))\big/ \eta_\omega
\end{split}
\end{equation}
Hence, the Proposition holds.
\end{proof}

\section{Lemma 1 and Proof of Proposition 4}
\subsection{Lemma 1 and a proof}
\begin{Lem}
Suppose that $(\mathbb{L}_1^{a},\mathbb{L}_2^{a})$ and $(\mathbb{L}_1^{b},\mathbb{L}_2^{b})$ are two mutually isolated clusters of $(\mathbb{L}_1,\mathbb{L}_2)$.  There is no path between any label $\ell_a\in\mathbb{L}_1^{a}$ and any label $\ell_b\in\mathbb{L}_1^{b}$ in the undirected graph $G(V,E)$ constructed by $(\mathbb{L}_1,\mathbb{L}_2)$.
\end{Lem}
\begin{proof}
This proposition is proved by reductio. Suppose that there  exists a path between $\ell_a \in \mathbb{L}_{1}^{a}$ and $\ell_b\in \mathbb{L}_{1}^{b}$, denoted by an alternating sequence of
vertices and edges, $\ell_a,e_a,\ell_1,e_1,\cdots,\ell_K,e_K,\ell_b$. Three cases  which exhaust all possible paths are discussed as follow.

\begin{itemize}
 \item All non-end vertices $\ell_1,\cdots,\ell_{K}$ belong to $\mathbb{L}_{1}^{a}$.  Based on the definition of path, $e_K=(\ell_{K},\ell_b)$ is an edge. That is to say, $\ell_{K}$ has common associated hypothetic objects from sensor 2 with  hypothetic object $\ell_b$. Thus, there is a hypothetic object $\ell'\in\mathbb{L}_{2}^{b}$  such that,
\begin{equation}
d(\ell_K,\ell')\leq \gamma,
\end{equation}
or there is a hypothetic object $\ell'\in\mathbb{L}_2^a$ such that
\begin{equation}
d(\ell_b,\ell')\leq \gamma
\end{equation}
 \item All non-end vertices $\ell_1,\cdots,\ell_{K}$ belong to $\mathbb{L}_{1}^b$. Similarly, $e_a=(\ell_a,\ell_1)$ is an edge, and thus there is a hypothetic object $\ell'\in\mathbb{L}_{2}^{b}$  such that
\begin{equation}
d(\ell_a,\ell')\leq \gamma,
\end{equation}
or or there is a hypothetic object $\ell'\in\mathbb{L}_2^a$ such that
\begin{equation}
d(\ell_1,\ell')\leq \gamma.
\end{equation}
 \item A part of  non-end vertices $\ell_1,\cdots,\ell_{K}$ belong to $\mathbb{L}_{1}^{a}$, while the other part $\ell_{K_0+1},\cdots,\ell_{K}$ belong to $\mathbb{L}_{1}^{b}$. Similarly, $(\ell_{K_0},\ell_{{K_0}+1})$ is  an edge.   and thus  there is a hypothetic object $\ell'\in\mathbb{L}_{2}^{b}$  such that
\begin{equation}
d(\ell_{K_0},\ell')\leq \gamma,
\end{equation}
or there is a hypothetic object $\ell'\in\mathbb{L}_{2}^{a}$  such that
\begin{equation}
d(\ell_{K_0+1},\ell')\leq \gamma.
\end{equation}
\end{itemize}
 To summary, all the aforementioned cases are inconsistent with that $(\mathbb{L}_1^{a},\mathbb{L}_2^{a})$ and $(\mathbb{L}_1^{b},\mathbb{L}_2^{b}) $ are isolated clusters. Consequently, there does not exist a path between $\ell_a$ and $\ell_b$, and the proposition holds.
\end{proof}
\subsection{A Proof of Proposition 4}
\begin{proof}
Firstly, we focus on proving that $C$ is an isolated clustering.
As presented in Section IV-C 2), the union $C_{\text{I}}\cup C_\text{II}$ in the finalized clustering $C$ of (\ref{type1-2}) is obtained through seeking the connected components of the constructed undirected graph $G(V,E)$.  Recall the structure of the undirected graph, each vertex $\ell$ represents a  hypothetic object in $\mathbb{L}_1$, and the edge  means that two hypothetic objects  of sensor 1 have  common associated hypothetic objects from sensor 2.

 Based on the definition of the connected component, for any  $g_1\neq g_2\in\{g:\mathcal{C}_{g}\in C_{\text{I}}\cup C_{\text{II}}\}$, any vertex in $V_{g_1}$ is not connected with any vertex in $V_{g_2}$, and thus cannot be paired to be an edge. Hence, any hypothetic object $\ell_1\in \mathbb{L}_{1,_{g_1}}(\mathbb{L}_{1,{g_1}}=V_{g_1})$ and any hypothetic object $\ell_2\in \mathbb{L}_{1,g_2} (\mathbb{L}_{1,g_2}=V_{g_2})$ have no common associated hypothetic objects from sensor 2. That is to say that for any two different clusters $\mathcal{C}_{g_1}$ and $\mathcal{C}_{g_2}$ of $C_{\text{I}}\cup C_{\text{II}}$, the following holds:
\begin{equation}
\mathbb{L}_{2,g_1}\cap \mathbb{L}_{2,g_2}=\emptyset.
\end{equation}

 In addition,  for any cluster $\mathcal{C}_{g_3}\in C_{\text{III}}$ where $C_{\text{III}}$ is given in (\ref{type3}), we have $\mathbb{L}_{1,g_3}$ is an empty set and $\mathbb{L}_{2,g_3}\subseteq\mathbb{L}_2\backslash (\cup_{\ell\in\mathbb{L}_1}\Psi_2^{(\ell)})$.  Then, based on (\ref{cluster_1}) and (\ref{union-find-clustering}), we have \begin{equation}\label{C-III-L2}
 \mathbb{L}_{2,g_3}\cap\left(\bigcup_{g: \mathcal{C}_g\in C_{\text{I}}\cup C_{\text{II}}}\mathbb{L}_{2,g}\right)=\emptyset.
  \end{equation}

As a result, for any two different clusters $\mathcal{C}_g\neq \mathcal{C}_{g'}\in \left(C_{\text{I}}\cup C_{\text{II}}\cup C_{\text{III}}\right)$:
\begin{itemize}
\item
 if $\mathbb{L}_{1,g}\times\mathbb{L}_{2,g'}\neq \emptyset $ any hypothetic object $\ell$ in $\mathbb{L}_{1,g}$ of the $g$th cluster and  any $\ell'$ in $\mathbb{L}_{2,g'}$ of the $g'$th cluster satisfy
\begin{equation}
d(\ell,\ell')>\gamma;
\end{equation}
\item if $\mathbb{L}_{1,g'}\times\mathbb{L}_{2,g}\neq \emptyset$, any hypothetic object $\ell$ in $\mathbb{L}_{1,g'}$ of the $g'$th cluster and  any $\ell'$ in $\mathbb{L}_{2,g}$ of the $g$th cluster satisfy
\begin{equation}
d(\ell,\ell')>\gamma.
\end{equation}
\end{itemize}
Thus, the following inequality holds:
\begin{equation}\label{clustering-criterion11}
\min_{(\ell,\ell')\in\mathcal{P}} d(\ell,\ell')>\gamma,
\end{equation}
where $\mathcal{P}$ is given in (\ref{clustering-criterion2}). According to Definition 5, $C$ is an isolated clustering.

Secondly, that $C$ is indivisible  is proved by reductio.  Suppose that any isolated cluster $\mathcal{C}_g\in C$ can be further divided into two isolated clusters $\mathcal{C}_{g}^{a}=(\mathbb{L}_{1,g}^{a},\mathbb{L}_{2,g}^{a})$ and $\mathcal{C}_{g}^{b}=(\mathbb{L}_{1,g}^{b},\mathbb{L}_{2,g}^{b})$, where
\begin{align}
\mathbb{L}_{s,g}^{a}=&\mathbb{L}_{s,g}\backslash\mathbb{L}_{s,g}^{b},s=1,2.
\end{align}

If $\mathcal{C}_g\in \left(C_{\text{I}}\cup C_{\text{I}}\right)$,  based on Lemma 1,  there does not exist a path between an arbitrary vertex $\ell_a\in\mathbb{L}_{1,g}^{a}$ and  an arbitrary vertex $\ell_b\in\mathbb{L}_{1,g}^{b}$. Then, the corresponding subgraph $G_g(V_g,E_g)$ is not a connected component, which is inconsistent with  that  all $G_1(V_1,E_1),\cdots,G_{N_G}(V_{N_G},E_{N_G})$ are  connected components of $G(V,E)$.

If $\mathcal{C}_g\in C_{\text{III}}$, since $\mathbb{L}_{1,g}=\emptyset$, we have $\mathbb{L}^a_{1,g}=\mathbb{L}^b_{1,g}=\emptyset$.  Based on Definition 4, $\mathbb{L}^a_{1,g}\cup\mathbb{L}^a_{2,g}\neq\emptyset$ and $\mathbb{L}^b_{1,g}\cup\mathbb{L}^b_{2,g}\neq\emptyset$, hence, neither $\mathbb{L}^a_{2,g}$ or $\mathbb{L}^b_{2,g}$ is an empty set. Also because the sets $\mathbb{L}^a_{2,g}$ and $\mathbb{L}^b_{2,g}$ are disjoint, then $\mathbb{L}_{2,g}=\mathbb{L}^a_{2,g}\cup\mathbb{L}^b_{2,g}$ is not a singleton set, which is inconsistent with clusters of type $C_{\text{III}}$.

Above all, the proposition holds.
\end{proof}

\bibliographystyle{IEEEtran}
\bibliography{GCI-MB-G}
\end{document}